\documentclass[letterpaper, 10pt, twocolumn, final]{IEEEtran}
\usepackage{fancyhdr}
\usepackage{amsmath,epsfig}
\usepackage{threeparttable}
\usepackage{epsf,epsfig}
\usepackage{amsthm}
\usepackage{amsmath}
\usepackage{amssymb}
\usepackage{amsfonts}
\usepackage{algorithmic}
\usepackage{algorithm}
\usepackage[noadjust]{cite}
\usepackage{color}
\usepackage{dsfont}
\usepackage{subfigure}
\usepackage{mathrsfs}

\newtheorem{theorem}{Theorem}

\newtheorem{example}{Example}

\newtheorem{lemma}{Lemma}

\newtheorem{proposition}{Proposition}

\newtheorem{assumption}{Assumption}

\def\E{\mathsf{E}}

\def\phi{\varphi}

\def\l{\left}
\def\r{\right}
\def\({\left(}
\def\){\right)}

\def\bff{{\mathbf{f}}}

\def\bh{{\mathbf{h}}}

\def\bq{{\mathbf{q}}}

\def\bs{{\mathbf{s}}}

\def\bu{{\mathbf{u}}}

\def\b0{{\mathbf{0}}}

\newcommand{\nn}{\nonumber}

\begin{document}

\title{\huge \setlength{\baselineskip}{30pt} Stochastic Control of Event-Driven Feedback  in Multi-Antenna  Interference  Channels}
\author{\large \setlength{\baselineskip}{15pt}Kaibin Huang, Vincent K. N. Lau, and Dongku Kim \thanks{Copyright (\copyright) 2011 IEEE. Personal use of this material is permitted. However, permission to use this material for any other purposes must be obtained from the IEEE by sending a request to pubs-permissions@ieee.org.

K. Huang and D. Kim are with Yonsei University, S.  Korea;  V. K. N. Lau is   with Hong Kong University of Science and Technology. Email: huangkb@me.com, eeknlau@ust.hk, dkkim@yonsei.ac.kr. This work was  supported by the National Research Foundation of Korea under the grant $2011$-$8$-$0740$, Hong Kong RGC under the grant $614910$, and Ministry of Knowledge Economy of Korea under the grant $10035389$. 
}}

\maketitle

\begin{abstract}
Spatial interference avoidance is a simple and effective way of mitigating interference in  multi-antenna wireless networks. The deployment of this technique requires channel-state information (CSI) feedback from each receiver to all interferers, resulting in substantial network overhead. To address this issue, this paper proposes the method of distributive control that intelligently allocates CSI bits over multiple feedback links and adapts feedback to channel dynamics. For symmetric channel distributions, it is optimal for each receiver to equally allocate  the average sum-feedback rate for different feedback links, thereby decoupling  their control.   Using the criterion of minimum sum-interference power,  the optimal feedback-control  policy is shown using stochastic-optimization theory to exhibit opportunism. Specifically, a specific feedback link is turned on only when the corresponding transmit-CSI error is  significant or  interference-channel gain   is large, and the optimal number of feedback bits  increases with this gain.   For high mobility and considering the  sphere-cap-quantized-CSI  model, the optimal feedback-control policy is shown to perform  \emph{water-filling} in time,  where the number of feedback bits increases logarithmically  with the corresponding interference-channel gain. Furthermore, we consider asymmetric channel distributions with heterogeneous path losses and high mobility, and prove the existence of a unique optimal policy for jointly controlling multiple feedback links. Given the sphere-cap-quantized-CSI model, this  policy is shown to perform water-filling over feedback links.  Finally, simulation demonstrates that feedback-control yields significant throughput gains compared with the conventional differential-feedback method. 
\end{abstract}


\section{Introduction}\label{Section:Intro}
Interference limits the performance of decentralized wireless networks but can be effectively mitigated by multi-antenna techniques, namely spatial interference cancelation and avoidance. 
In a frequency-division-duplexing  network, spatial interference avoidance at interferers requires feedback of interference channel state information (CSI)  from all interfered receivers, called \emph{cooperative feedback}. Given finite-rate cooperative feedback, CSI  quantization errors result in residual interference.  Suppressing such interference requires high-resolution feedback over a network of feedback links, resulting in overwhelming network overhead.  This calls for research on intelligent feedback control  that optimally allocates  feedback bits over multiple feedback links and adapts feedback to channel dynamics, which is the theme   of this paper. 

\subsection{Prior Work}
Extensive research has been carried out  on designing   feedback-CSI-quantization algorithms for multi-antenna systems, called   \emph{limited feedback} \cite{Love:OverviewLimitFbWirelssComm:2008},  based on different approaches    including  line packing \cite{LovHeaETAL:GrasBeamMultMult:Oct:03} and Lloyd's algorithm \cite{Lau:CapBlockFadingCardConstFeedback:2004}. Besides quantization, another effective approach   for compressing feedback CSI is to explore CSI redundancy due to the  wireless-channel correlation in time \cite{Huang:LimFbBeamformTemporallyCorrChan:2009, KimLov:MIMODiffFBSlowFading:2011}, frequency \cite{ChoiMondal:InterpPrecodSpaMuxMIMOOFDMLimFb:2006}, and space \cite{MondalHeath:AdaptiveQ:04}. Though a few feedback bits suffice in a point-to-point multi-antenna system, the feedback requirement is more stringent in multi-antenna downlink where CSI errors cause multiuser interference \cite{Jindal:MIMOBroadcastFiniteRateFeedback:06}. This motivates the joint design of CSI feedback and scheduling  algorithms to exploit multiuser diversity for reducing the required numbers of feedback bits
\cite{YooJindal:FiniteRateBroadcastMUDiv:2007, SharifHassibi:CapMIMOBroadcastPartSideInfo:Feb:05,Huang:OrthBeamSDMALimtFb:07,Huang:SDMASumFbRate:06}. Both high-resolution feedback for multi-antenna downlink and progressive feedback for correlated channels  require CSI feedback with adjustable  resolutions. This is realized using  hierarchical CSI-quantizer codebooks \cite{Kim:FeedbackSubspaceCoopAntennaTempCorr, TrivHuang:TXRXDesignDownlinkLimFb} or  systematic codebook generation \cite{Heath:ProgressiveRefineBeamform, RagHeath:SystematicCodebookDesignBeam}. The current work also concerns variable-rate feedback but focuses on feedback control rather than codebook designs.  

Recent research on limited feedback explores more complex  network topologies. In \cite{ThuBolcskei2009:IntefAlignLimitedFeedback}, the decentralized wireless networks based on interference alignment \cite{CadJafar:InterfAlignment:2007}  are considered,  and the required scaling of the numbers of feedback bits  with respect to the signal-to-noise ratio (SNR) is derived such that the channel capacity is achieved for high SNRs. The Grassmannian codebooks designed for point-to-point beamforming systems with limited feedback is shown in  \cite{KhoYu:GrassmannianBeamformRelay} to be suitable for multiple-input-multiple-output (MIMO) amplify-and-forward relay systems. The algorithms for  cooperative feedback from the primary user to the secondary user are designed in \cite{HuangZhang:CoopFeedbackCognitiveRadio} for implementing  cognitive beamforming in two-user cognitive-radio systems.  Moreover, Lloyd's algorithm is applied in \cite{BhagaHeath:JointCSIQuantMulticellCoop:2010} to jointly quantize the CSI sent by a mobile to the desired and interfering base stations.  The above prior work does not explicitly optimize the tradeoff between the network performance and the amount of CSI overhead.

In wireless networks, excessive CSI feedback yields marginal performance gain per additional feedback bit but insufficient feedback causes unacceptable performance degradation. Therefore, feedback control is a pertinent issue for designing efficient wireless networks. In \cite{Love:DuplexDistortLimFbMIMO:2006}, CSI feedback rates are optimized for maximizing the sum throughput in a two-way beamforming system where a pair of transceivers exchange both data and CSI.
For a transmit beamforming system,  bandwidth is optimally partitioned for CSI feedback and data transmission \cite{Xie:OptimalBWAllocDataFeedbackMISO:2006}. For point-to-point multi-antenna precoding, sub-optimal algorithms have been proposed to adjust the CSI-codebook size according to the channel state  \cite{SimonLeus:FBQuantLinearPrecodSpatialMultiplex:2008} or jointly with the feedback interval based on channel temporal correlation \cite{KimLove:FrequentOrInfreqFBBeamforming:2011}. 
The problem of splitting the sum-feedback rate by a mobile for  multiple cooperative-feedback links to interferers  is studied in \cite{RamyaHeath:AdaptiveBitPartionDelayedLimFb} in the context of base-station collaboration. It was shown that more feedback bits should be sent to nearer interfering base stations so as to reduce the throughput loss caused by feedback quantization. The splitting of the sum-feedback rate among multiple users in a multi-antenna downlink system was investigated in \cite{KhoYu:LimFBCodebookDesignMultiuserCase:2010}, where the optimal feedback rate for a user is shown to increase logarithmically with the target signal-to-interference-plus-noise ratio (SINR). 
The feedback-bit  allocation considered in prior work is mostly static, targeting dedicated feedback channels in cellular networks \cite{3GPP-LTE}. In decentralized networks where a feedback channel is shared by multiple users,  more efficient feedback-allocation  should be adapted to   channel dynamics, motivating the event-driven feedback and stochastic feedback control.

\subsection{Contributions and Organization}
This work adopts   the approach of stochastic feedback control proposed in \cite{HuangLau:EventDrivenFeedbackControlBeamforming} but targets more complex systems.  Specifically, this paper concerns  the $K$-user multiple-input-single-output (MISO) interference channel where
there is an event-driven feedback controller at each receiver. The feedback controller dynamically and distributively determines the CSI feedback rate for each feedback link according to local CSI. As a result, each feedback controller serves multiple cooperative-feedback links in the current  system rather than a single feedback link to the intended transmitter  as in \cite{HuangLau:EventDrivenFeedbackControlBeamforming}. \footnote{In this paper, we focus on  cooperative feedback with some discussion of direct-link feedback, namely CSI feedback from receivers to their intended transmitters. Hereafter, cooperative feedback is referred to simply as feedback whenever there is no confusion.} Furthermore, we generalize the on/off feedback control in \cite{HuangLau:EventDrivenFeedbackControlBeamforming} to the variable-rate feedback control. 

This work establishes a novel approach of using stochastic feedback control to  achieve the  optimal tradeoff between the CSI-feedback overhead and sum interference power in the $K$-user multi-antenna interference channel.  The feedback controllers are designed based on several key assumptions. Channel coefficients are assumed to be independent and identically distributed (i.i.d.). The expectation of a CSI quantization error is assumed to be a monotone decreasing and convex function of the number of feedback bits, which is consistent with the popular CSI-quantizer models in \cite{Zhou:QuantifyPowrLossTxBeamFiniteRateFb:2005,YeungLove:RandomVQBeamf:05}. Moreover,  the channel parameters, namely channel gains and transmit CSI (CSIT) errors, are assumed to vary in time following Markov chains. The channel temporal correlation is further characterized by two assumptions. Given no feedback, samples of the channel-parameter processes conditioned on large past realizations stochastically dominate those conditioned on small ones; given feedback, the tail probability of the CSIT error is a monotone decreasing and convex function of the corresponding number of feedback bits in the past slot. The channels are assumed to follow independent block fading for the limiting case of high mobility. Based on these assumptions, the key findings of this work are summarized as follows. 
\begin{itemize}

\item[--]  Under an average sum-feedback-rate constraint,  a feedback  controller is designed  as a Markov decision process with average cost. By channel symmetry, it is optimal for each controller to equally split the average sum-feedback rate for all feedback links, reducing 
the problem of optimizing the multiple-feedback-link control policy  to the single-feedback-link-policy optimization. The optimal policy for minimizing the average sum-interference power is shown to exhibit opportunism. Specifically, feedback should be performed only when the corresponding interference-channel gain is large or the CSIT error is significant. Upon feedback, the optimal number of feedback bits for each feedback link increases with the corresponding interference-channel gain but is independent with the observed CSIT error. 

\item[--] For  high mobility and considering the sphere-cap-quantized-CSI model \cite{YooJindal:FiniteRateBroadcastMUDiv:2007, Zhou:QuantifyPowrLossTxBeamFiniteRateFb:2005},  more elaborate properties of the optimal feedback-control policy are derived. Specifically, it is shown that the number of feedback bits for each feedback link follows water-filling in time and is proportional to the logarithm of the corresponding interference-channel gain.

\item[--] We also consider asymmetric channel distributions where interference-channel gains are scaled by heterogeneous path losses. For high mobility, the problem of feedback-control-policy optimization  is decomposed into a \emph{master problem} that optimally allocates average feedback rates for multiple feedback links, and a \emph{sub-problem} that optimizes the  policy for controlling the feedback-bit allocation in time for a particular feedback link given an allocated average feedback rate. This decomposed optimization problems are proved to yield  a unique optimal policy. Furthermore, given the sphere-cap-quantized-CSI model,  the optimal feedback-control policy is shown to perform  water-filling over feedback links. 
\end{itemize}

The remainder of this paper is organized as follows. The system model is described in Section~\ref{Section:System}. The problem formulation for the optimal feedback control is presented in Section~\ref{Section:ProbForm}. The optimal feedback-control policies for the general case and the limiting case of high mobility are analyzed in Section~\ref{Section:FBControl:LowMob} and \ref{Section:FBControl:HiMob}, respectively. In~Section~\ref{Section:AsymChan},  the design of the feedback controller for asymmetric channel distributions is discussed. Simulation results are presented in Section~\ref{Section:Simulation}.
  
\begin{figure}[t]
\begin{center}
\includegraphics[width=8cm]{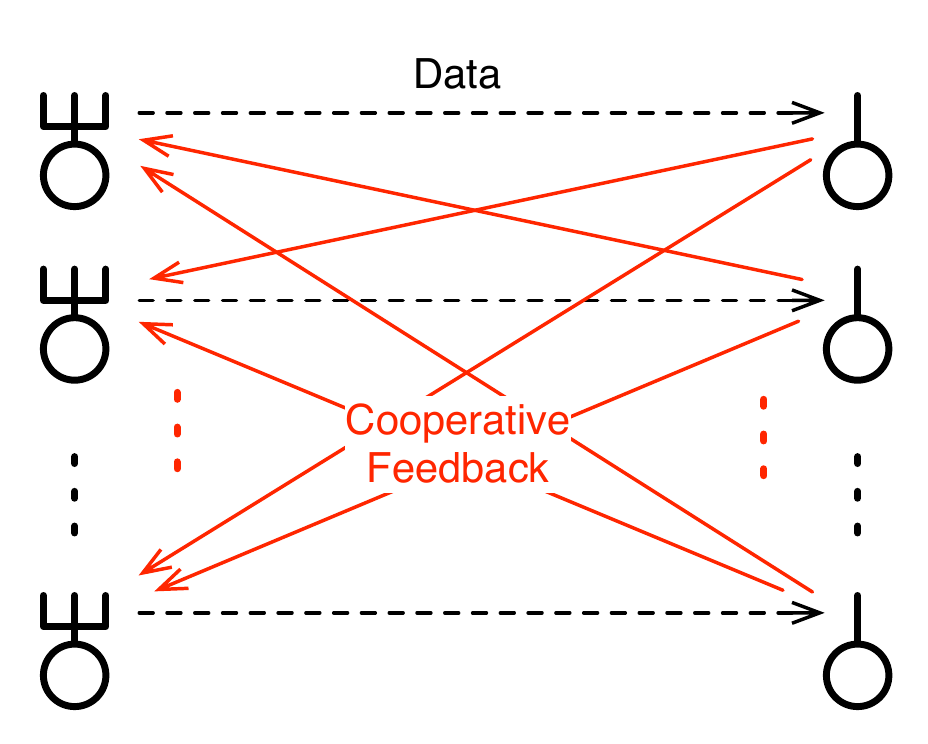}
\caption{The $K$-user MISO interference channel with cooperative feedback }
\label{Fig:Sys}
\end{center}
\end{figure}

\section{System Model}\label{Section:System}
We consider the  $K$-user MISO interference channel as illustrated in Fig.~\ref{Fig:Sys}.  Provisioned with $L$ antennas, each transmitter sends a single data stream to an intended receiver using beamforming.   
As illustrated in Fig.~\ref{Fig:FbControl},  time is slotted and each slot is divided into  the \emph{feedback phase} (feedback control   and cooperative CSI feedback) and the \emph{data phase} (data transmission).  Each system parameter affected by feedback is represented by the same symbol without and with the accent $``\ \check{}\ "$, corresponding to the beginnings of the feedback and data phases, respectively. Moreover, the subscript $t$ denotes the slot index.

\subsection{Zero-Forcing Transmit Beamforming}
Each transmitter uses beamforming to null interference to  $(K-1)$ unintended receivers. Let $\bh^{[mn]}_t$ denote the $L\times 1$ vector representing the channel from transmitter $n$ to   receiver $m$. To facilitate exposition, we decompose $\bh^{[mn]}_t$ as $\bh^{[mn]}_t =\sqrt{g_t^{[mn]}} \bs_t^{[mn]}$ where $g_t^{[mn]} = \|\bh^{[mn]}_t\|^2$ is the channel  gain and $\bs_t^{[mn]} = \bh^{[mn]}_t/\|\bh^{[mn]}_t\|$ specifies the channel direction. Transmitter $n$ applies zero-forcing beamforming by choosing  its beamformer $\bff^{[n]}_t$ to be orthogonal to the interference-channel directions. As a result, $K$ links are decoupled if all transmitters have perfect CSIT of the channels to their interfered receivers.

Consider the scenario where transmit beamforming at a transmitter  relies on finite-rate CSI feedback from interfered receivers. Let $\bu^{[mn]}_t$ with unit norm denote the CSIT at transmitter $n$ updated by the feedback of $\bs^{[mn]}$ from receiver $m$. Then the zero-forcing  beamformer $\bff^{[n]}_t$ at transmitter $n$ satisfies the constraints: 
 $(\bff^{[n]}_t)^\dagger \bu^{[mn]}_t=0$ for all   $m\neq n$, which requires $L \geq K$.  Under the finite-rate feedback constraints, imperfect CSIT results  in residual interference between links. The interference from transmitter $n$ to receiver $m$ has the power 
\begin{equation}\label{Eq:InterfPwr:Cmp}
I^{[mn]}_t = g^{[mn]}_t\l|(\bff^{[n]}_t)^\dagger \bs^{[mn]}_t \r|^2, \quad m \neq n 
\end{equation}
where unit  transmission power is used by all transmitters. 

\begin{figure}[t]
\begin{center}
\includegraphics[width=9cm]{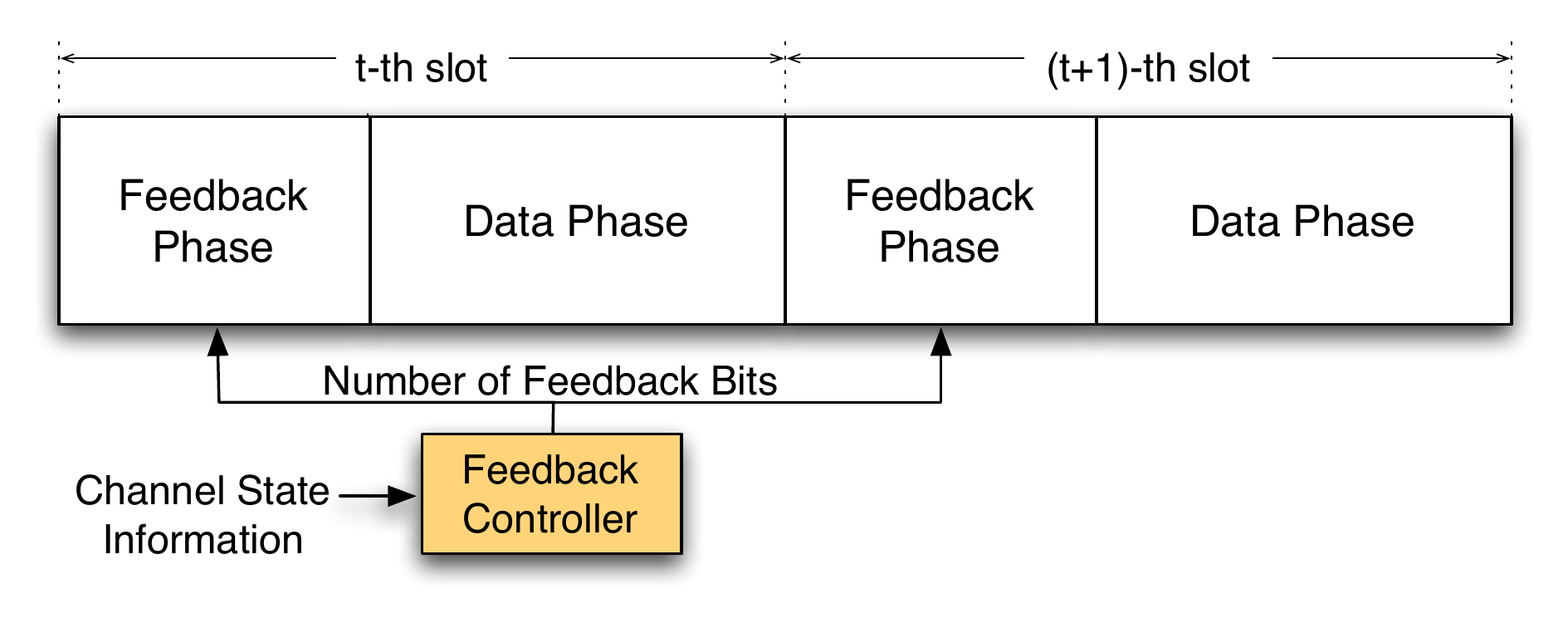}
\caption{Variable-rate feedback control }
\label{Fig:FbControl}
\end{center}
\end{figure}

\subsection{Variable-Rate Feedback Control}
In the feedback phase of  every slot, each receiver, say receiver $m$,  sends the quantized version $\hat{\bs}^{[mn]}_t$ of   $\bs^{[mn]}_t$  to interferer  $n$  in a variable-length packet comprising $B^{[mn]}$ bits. The variable-rate feedback is modeled as $B^{[mn]}\in\mathds{B}$, where $\mathds{B}$ is a set of nonnegative integers including $0$ that corresponds to no feedback.  As illustrated in Fig.~\ref{Fig:FbControl}, a feedback controller at each receiver controls the number of feedback bits sent to a particular interferer by observing the  interference-channel gain and the CSIT error that is defined as follows. The dynamics of the CSIT $\bu^{[mn]}_t$ at transmitter $n$ can be specified as 
\begin{eqnarray}
\bu^{[mn]}_{t+1}&=&\check{\bu}^{[mn]}_t\nn\\
& =& \l\{\begin{aligned}
&\hat{\bs}^{[mn]}_t,&& B^{[mn]}_t > 0\\
&\bu^{[mn]}_{t},&& B^{[mn]}_t = 0. 
\end{aligned}
\r. \label{Eq:TxCSI:a}
\end{eqnarray}
 The CSIT error is  defined as $\delta^{[mn]}_t = 1- \l|(\bs^{[mn]}_t)^\dagger \bu^{[mn]}_t \r|^2$ with  $\delta^{[mn]}_t = 0$ for the case of perfect CSIT: $\bu^{[mn]}_{t} = \bs^{[mn]}_{t}$ \cite{LovHeaETAL:GrasBeamMultMult:Oct:03}.
The feedback controller at receiver $m$ observes the state $\l\{\l(g^{[mn]}_t, \delta^{[mn]}_t \r)\mid n\neq m\r\}$ and generates the feedback decision $\l\{B^{[mn]}_t\mid n\neq m\r\}$. 
Similarly, we define the CSI-quantization error as $\epsilon^{[mn]}_t =  1- \l|(\bs^{[mn]}_t)^\dagger \hat{\bs}^{[mn]}_t \r|^2$. 
\begin{assumption} \label{As:QErr} The conditional expectation $\E\l[\epsilon^{[mn]}_t\mid B^{[mn]}_t\r]$  is a monotone decreasing and convex function of  $B^{[mn]}_t$. 
\end{assumption}

\begin{example}[Sphere-cap-quantized-CSI model] \label{Ex:SphereCap} \emph{The quantization error $\epsilon^{[mn]}$ is modeled in 
\cite{YooJindal:FiniteRateBroadcastMUDiv:2007, Zhou:QuantifyPowrLossTxBeamFiniteRateFb:2005}  to be uniformly distributed on a sphere-cap in $\mathds{C}^L$ with the following distribution function
\begin{equation}\label{Eq:EpsCCDF:1}
\Pr\l(\epsilon^{[mn]}\leq \tau\mid B^{[mn]}\r) =  \left\{\begin{aligned}
&2^{B^{[mn]}} \tau^{L-1},&& 0 \leq \tau \leq 2^{-\frac{B^{[mn]}}{L-1}}\\
&1,&& \textrm{otherwise}. 
\end{aligned}\right.
\end{equation}
Using this model,  the expectation of $\epsilon^{[mn]}$ is obtained as 
\begin{equation}
\E\l[\epsilon^{[mn]}\mid B^{[mn]}\r] = \frac{L-1}{L} 2^{-\frac{B^{[mn]}}{L-1}}  \label{Eq:ExpErr:SCap}
\end{equation}
which is a monotone decreasing and convex function of $B^{[mn]}$, consistent with Assumption~\ref{As:QErr}. }
\end{example}

\begin{example}[Random-vector quantization] \label{Ex:RVQ} \emph{As shown in \cite{YeungLove:RandomVQBeamf:05}, the use of a random beamformer codebook of i.i.d. and isotropic unitary vectors results in the following distribution of $\epsilon^{[mn]}$ 
\begin{equation}\label{Eq:EpsCCDF:2}
\Pr\l(\epsilon^{[mn]} \geq \tau \mid B^{[mn]}\r) = (1-\tau^{L-1})^{2^{B^{[mn]}}}
\end{equation}
and the expectation 
\begin{eqnarray}
\E\l[\epsilon^{[mn]} \mid B^{[mn]}\r] &=& 2^{B^{[mn]}}\mathsf{beta}\l(2^{B^{[mn]}}, \frac{L}{L-1}\r)\nn\\
&\approx & a e^{-\frac{1}{L-1}B^{[mn]}}, \quad B^{[mn]} \gg 1  \label{Eq:ExpErr:RVQ}
\end{eqnarray}
where $\mathsf{beta}(\cdot, \cdot)$ denotes the beta function and $a$ is a constant. The last  expression  is a monotone decreasing and convex function of $B^{[mn]}$, justifying Assumption~\ref{As:QErr}. }
\end{example}

Next, it follows from \eqref{Eq:TxCSI:a} that 
\begin{equation}\label{Eq:TxCSI}
\check{\delta}^{[mn]}_t = \l\{\begin{aligned}
&\epsilon^{[mn]}_t,&& B^{[mn]}_t > 0\\
&\delta^{[mn]}_t,&& B^{[mn]}_t = 0 
\end{aligned}
\r.
\end{equation}
given that channels remain constant within each slot as assumed in the sequel. Note that $\check{g}^{[mn]} = g^{[mn]}$ since it is unaffected by feedback. 

{ Finally, it is important to note that besides CSI, controlled feedback requires addition bits for specifying the number of feedback-CSI bits since it varies with the channel state.  Such overhead is unnecessary for feedback schemes with fixed numbers of feedback bits (see e.g., \cite{LovHeaETAL:GrasBeamMultMult:Oct:03, KimLov:MIMODiffFBSlowFading:2011}). 
Let $D$ denote the number of available decisions for the feedback-control policy. Assuming that the policy is known to the transmitter, the total number of feedback bits from receiver $m$ to transmitter $n$ in the $t$-th slot is $B^{[mn]}_t +\l\lceil\log_2 D \r\rceil$. It is observed from simulation that $D$ for the optimal policy is relatively small e.g., $3$ or $4$ (see Fig.~\ref{Fig:Policy}).} 

\subsection{Channel Model}
Channels vary with time but remain constant within each slot.  For simplicity, all  channel coefficients, namely the elements of the vectors $\l\{ \bh^{[mn]}_t\r\}$,  are assumed to be  samples of i.i.d circularly-symmetric complex Gaussian processes with unit variance, which is denoted as $\mathcal{CN}(0,1)$ (asymmetric channel distributions are considered in Section~\ref{Section:AsymChan}).   Note that as a result of channel isotropicity, the two channel parameters $g_t^{[mn]}$ and $\delta_t^{[mn]}$ are independent conditioned on $B_t^{[mn]}$.  The channel temporal correlation  is modeled using the following  two assumptions. 

\begin{assumption}\label{AS:SD} Each channel coefficient evolves as a Markov chain. Given $B^{[mn]}_{t-1}=0$, the distributions of $\l(g_t^{[mn]}, \delta_t^{[mn]}\r)$  conditioned on $\l(g_{t-1}^{[mn]}, \delta_{t-1}^{[mn]}\r)$  satisfy 
\begin{align}
&\Pr\l(\delta^{[mn]}_t\geq \tau_1\!\mid\! \delta_{t-1}^{[mn]} = a_1\r)\!\geq\! \Pr\l(\delta^{[mn]}_t\geq \tau_1\!\mid\! \delta_{t-1}^{[mn]} = b_1\r) \nn\\
&\Pr\l(g^{[mn]}_t\geq \tau_2 \!\mid\! g_{t-1}^{[mn]} = a_2\r)
 \!\geq\! \Pr\l(g_t^{[mn]}\geq \tau_2 \!\mid\! g_{t-1}^{[mn]} = b_2\r)\nn
\end{align}
if $a_1\geq b_1$ and $a_2\geq b_2$, where $0 \leq \tau_1 \leq 1$ and $\tau_2 \geq 0$. 
\end{assumption}
\noindent The above assumption states that given no feedback, large CSIT error and channel power in the current slot are likely to stay large in the next  slot due to channel temporal correlation.

\begin{assumption}\label{AS:EspCorr:Convex} For $B^{[mn]}_{t-1}>0$, the  conditional distribution $\Pr\l(\delta^{[mn]}_{t}\geq \tau\mid B^{[mn]}_{t-1}\r)$   is a monotone decreasing and convex function of $B^{[mn]}_{t-1}$. 
\end{assumption}
\noindent Note that upon feedback,  $\delta^{[mn]}_{t}$ is independent of $\delta^{[mn]}_{t-1}$ as a result of  \eqref{Eq:TxCSI}. 

Finally, for the limiting case of high mobility, channels are assumed to follow independent block-fading channels. For this case, Assumption~\ref{AS:SD} and \ref{AS:EspCorr:Convex} are trivial and not required in the analysis.

\subsection{Performance Metric}
The objective for designing the distributed feedback controller  at each receiver is to minimize the average  interference power. For receiver $m$, this metric is given as 
\begin{equation}\label{Eq:Metric}
\bar{I}^{[m]} =  \lim_{T\rightarrow\infty}\frac{1}{T}\E\l[\sum_{t=1}^T\sum_{\substack{n=1\\n \neq m}}^KI_t^{[mn]}\r]
\end{equation}
with $I_t^{[mn]}$ given in \eqref{Eq:InterfPwr:Cmp}. Minimizing $\bar{I}^{[m]} $ suppresses the system performance degradation  caused by  quantizing feedback CSI e.g.,  the throughput loss  in  the following example.  

\begin{example}\emph{ Let $S^{[m]}_t$ and $I^{[m]}_t$ denote the signal and interference power received  at receiver $m$ in slot $t$, respectively. Assuming Gaussian signaling and high mobility, the throughput loss of the $m$-th data link is given as \cite{Jindal:MIMOBroadcastFiniteRateFeedback:06}
\begin{align}
\Delta R &= \E\l[ \log_2\l(1 + \frac{S^{[m]}_t}{\sigma^2 }\r)  - \log_2\l(1 + \frac{S^{[m]}_t}{\sigma^2 + I_t^{[m]}}\r) \r]\nn\\
&\leq \E\l [\log_2\l(1 + \frac{I_t^{[m]}}{\sigma^2}\r) \r]\nn\\
&\leq  \log_2\l(1 + \frac{\bar{I}^{[m]}}{\sigma^2}\r) \label{Eq:RateLoss:Ub}
\end{align}
where $\sigma^2$ is the variance of a sample of the  additive-white-Gaussian-noise process and \eqref{Eq:RateLoss:Ub} uses Jensen's inequality. It can be observed from \eqref{Eq:RateLoss:Ub} that minimizing an upper bound on the throughput loss is equivalent to minimizing $\bar{I}^{[m]}$. }
\end{example}


\section{Problem Formulation}\label{Section:ProbForm}
The design of the feedback controller is formulated as a stochastic optimization problem under  an average  sum-feedback constraint. 

The cost function and state space for feedback control are defined as follows. To this end, the channel shape $\bs^{[mn]}_t$ is decomposed as
\begin{align}
\bs^{[mn]}_t &= \sqrt{1- \check{\delta}_t^{[mn]}}  \check{\bu}^{[mn]}_t + \sqrt{ \check{\delta}_t^{[mn]}} \check{\bq}^{[mn]}_t\label{Eq:S:Decompose}\\
&= \sqrt{1- \delta_t^{[mn]}}  \bu^{[mn]}_t + \sqrt{\delta_t^{[mn]}} \bq^{[mn]}_t\nn
\end{align}
where $\bq^{[mn]}_t$ and $\check{\bq}^{[mn]}_t$ are unitary vectors orthogonal to $\bu^{[mn]}_t $ and $\check{\bu}^{[mn]}_t$, respectively. Based on the above decomposition, we can define the channel parameters  $\beta_t^{[mn]} = \l| (\bff^{[n]}_t)^\dagger  \bq^{[mn]}_t\r|^2$ and $\check{\beta}_t^{[mn]} = \l|(\bff^{[n]}_t)^\dagger  \check{\bq}^{[mn]}_t\r|^2$. Using these parameters and substituting \eqref{Eq:InterfPwr:Cmp} allow $\bar{I}^{[m]}$ in \eqref{Eq:Metric} to be written as 
\begin{equation}\label{Eq:SumIntf:PerRX}
\! \bar{I}^{[m]} = \lim_{T\rightarrow\infty}\frac{1}{T}\E\l[\sum_{t=1}^T\sum_{n\neq m}\E\l[g_t^{[mn]}\check{\beta}_t^{[mn]}\check{\delta}_t^{[mn]}\mid x^{[m]}_t, \mathbf{B}_t^{[m]} \r]\r]
\end{equation}
where $x^{[m]}_t$ and $\mathbf{B}_t^{[m]} = \l\{B_t^{[mn]}\mid n \neq m\r\}$ denote the state and decision of the feedback controller at receiver $m$, respectively. From \eqref{Eq:SumIntf:PerRX}, the controller's state should be intuitively chosen to comprise all channel parameters $\l\{\beta^{[mn]}_t, g^{[mn]}_t, \delta^{[mn]}_t \mid n \neq m\r\}$. 
However, this results in the coupling of  feedback control  at different receivers. Specifically, by definition, the state parameter $\beta_t^{[mn]}$ depends on the beamformer $\bff^{[n]}_t$ that in turn is computed based the feedback CSI from the receivers $\{m\mid m \neq n\}$,  and each of these receivers also controls other  beamformers. Therefore, to enable distributive feedback control, $\beta_t^{[mn]}$ is excluded from the controller's state and hence $x^{[m]}_t = 
\l\{g^{[mn]}_t, \delta^{[mn]}_t\mid n \neq m\r\}$ where each parameter pair $(g_t^{[mn]}, \delta_t^{[mn]})$ depends only on the single channel $\bh^{[mn]}_t$. Since all  channel vectors are  isotropic and that feedback control is  independent of $\l\{\beta_t^{[mn]}\r\}$, $\beta_t^{[mn]}$ and $\check{\beta}_t^{[mn]}$ can be shown to be 
beta$(1, L-2)$ random variables and independent with $(g_t^{[mn]}, \delta_t^{[mn]})$ \cite{YooJindal:FiniteRateBroadcastMUDiv:2007}. This simplifies \eqref{Eq:SumIntf:PerRX} as 
\begin{equation}
\begin{aligned}
\bar{I}^{[m]} =&\frac{1}{L-1}\lim_{T\rightarrow\infty}\frac{1}{T}\times\\
& \E\l[\sum_{t=1}^T\sum_{n\neq m }\E\l[g_t^{[mn]}\check{\delta}_t^{[mn]}\mid x^{[m]}_t, \mathbf{B}_t^{[m]}\r]\r]. 
\end{aligned}\label{Eq:SumIntf:PerRX:a}
\end{equation}

Consider a stationary feedback-control policy and  an average sum-feedback constraint where the total average feedback rate for each receiver is no more than $\bar{b}>0$. The optimal policy $\mathcal{P}_m^\star: x^{[m]}_t\rightarrow\mathbf{B}^{[m]}_t$ at receiver $m$ solves the following infinite-horizon stochastic optimization problem: \footnote{It is also possible to formulate the optimal feedback control as a finite-horizon stochastic optimization problem. However, the current infinite-horizon formulation not only leads to a stationary control policy but also allows tractable analysis of the policy structure. Furthermore, the infinite-horizon approximation is justified by that a communication session in  a practical system such as 3GPP LTE usually spans over thousands of frames.} 

\begin{equation} \label{Eq:OpProb:AvFb:a}
\begin{aligned}
&\textrm{minimize:}&& \bar{I}^{[m]}(\mathcal{P}_m) \\
&\textrm{subject to :}&& \lim_{T\rightarrow\infty}\frac{1}{T}\E\l[\sum_{t=1}^T\sum_{ n\neq m } B^{[mn]}_t\r]\leq \bar{b}.
\end{aligned}
\end{equation}
Due to symmetric channel distributions, it is optimal for receiver $m$ to equally split $\bar{b}$ for $(K-1)$ feedback links. 
Consequently, the optimization of $\mathcal{P}_m$ reduces to that of the policy $\mathcal{P}$ for controlling an arbitrary single feedback link. To simplify notation, define the random process $(g_t, \delta_t, \check{\delta}_t, B_t)\sim (g_t^{[mn]}, \delta_t^{[mn]}, \check{\delta}_t^{[mn]}, B^{[mn]}_t)$ where ``$\sim$" represents equality in distribution, and the metric 
\begin{equation}
J = \lim_{T\rightarrow\infty}\frac{1}{T}\sum_{t=1}^T\E\l[g_t\check{\delta}_t\mid x_t, B_t\r] 
\end{equation}
where $x_t=(g_t, \delta_t)$ is the state of a single-feedback-link controller. 
Then $\mathcal{P}:x_t \rightarrow B_t$ can be designed by solving the following optimization problem: 
\begin{equation} \label{Eq:OpProb:AvFb}
\begin{aligned}
&\textrm{minimize:}&& J(\mathcal{P}) \\
&\textrm{subject to :}&& \lim_{T\rightarrow\infty}\frac{1}{T}\E\l[\sum_{t=1}^TB_t\r]\leq \frac{\bar{b}}{K-1}. 
\end{aligned}
\end{equation}

{
\section{The Optimal Feedback-Control Policy}\label{Section:FBControl:LowMob} In this section, we derive the optimal feedback-control policy for general mobility.  }
Given channel Markovity, the optimization problem in  \eqref{Eq:OpProb:AvFb} can be transformed into a stochastic optimization problem as follows.  By applying Lagrangian-multiplier theory, there exists a Lagrangian multiplier $\lambda > 0$ such that the optimal  policy $\mathcal{P}^\star$ that solves \eqref{Eq:OpProb:AvFb} also minimizes the following Lagrangian function: 
\begin{equation}\label{Eq:Lagrangian}
\mathcal{L}(\mathcal{P}) = \lim_{T\rightarrow\infty}\frac{1}{T}\E\l[\sum_{t=1}^T \l(\E\l[g_t\check{\delta}_t\mid x_t, B_t\r] + \lambda B_t\r)\r]. 
\end{equation}
Minizing $\mathcal{L}(\mathcal{P})$ is an average-cost stochastic optimization problem with a continuous state space. Though there exists no systematic method for solving this problem, it can be approximated by a discrete-space counterpart whose solution can be  computed efficiently using \emph{dynamic programming} \cite{Bertsekas07:DynamicProg}. The required state-space discretization is discussed and the resultant  optimal feedback-control policy analyzed in the following subsections.

\subsection{State-Space Discretization}
The spaces of the feedback-controller's state parameters $g_t$ and $\delta_t$ are discretized separately.  The set $\mathcal{G} = \{g_t \geq 0\}$ is partitioned into $M$ line segments  
$[\tilde{g}_1, \tilde{g}_2)$, $[\tilde{g}_2, \tilde{g}_3)$, $\cdots$, $[\tilde{g}_{M}, \infty)$ with $\tilde{g}_1=0$ and $0< \tilde{g}_1 < \tilde{g}_2<\cdots <\tilde{g}_{M}$.  These line segments are represented by a set of $M$ grid points $\hat{\mathcal{G}} = \{\bar{g}_m\}$ with $\bar{g}_m \in [\tilde{g}_{m}, \tilde{g}_{m+1})$. Specifically, $g_t \in \mathcal{G}$ is mapped to $\bar{g}_m$ if $g_t$ lies in the $m$-th line segment. Similarly, we divide the set $\mathcal{D} = \{0\leq \delta_t\leq 1\}$ into $N$ line segments $[\tilde{\delta}_1, \tilde{\delta}_2)$, $[\tilde{\delta}_2, \tilde{\delta}_3)$, $\cdots$, $[\tilde{\delta}_{N}, 1]$ with $\tilde{\delta}_1=0$ and $0< \tilde{\delta}_1 < \tilde{\delta}_2<\cdots <\tilde{\delta}_{N} < 1$ and represent these segments using a set of $N$ grid points $\hat{\mathcal{D}} = \{\bar{\delta}_n\}$ with $\bar{\delta}_n\in [\tilde{\delta}_n, \tilde{\delta}_{n+1})$. The optimization of the grid points $\hat{\mathcal{G}}$ and $\hat{\mathcal{D}}$  is outside the scope of this paper. Last, the discrete state space is represented by $\hat{\mathcal{X}} = \hat{\mathcal{G}}\times \hat{\mathcal{D}}$. 

The discretized version of the controller state $x_t$   is denoted as $\hat{x}_t = \{\hat{g}, \hat{\delta}_t\}$. Given Assumption~\ref{AS:SD}, $\{\hat{g}_t\}$ and $\{\hat{\delta}_t\}$ are two Markov chains whose transition probabilities are obtained as follows. Let  $P_{n, \ell}(B)$ denote the probability for the transition of $\hat{\delta}$ from the state $n$ to $\ell$ given the feedback decision $B$.  Then $P_{n, \ell}$ can be written as 
\begin{equation}\label{Eq:TxProb:z}
P_{n, \ell}(B) = \Pr(\hat{\delta}_{t+1} = \bar{\delta}_\ell \mid \hat{\delta}_t = \bar{\delta}_n,   B_t = B)\nn
\end{equation}
where $1\leq n, \ell\leq N$. 
Similarly, let  $\tilde{P}_{m, k}$ denote the transition probability  for $\{\hat{g}_t\}$, which is given as
\begin{equation}
\tilde{P}_{m, k} = \Pr(\hat{g}_{t+1}= \bar{g}_k \mid \hat{g}_t =\bar{g}_m)\nn
\end{equation}
where $1\leq m, k\leq M$. Note that given $B$, $P_{n, \ell}(B)$ and $\tilde{P}_{m, k}$ are independent as a result of channel isotropicity.   Last, the transition kernel for the controller-state Markov chain $\{\hat{x}_t\}$ is $\{\tilde{P}_{m, k}\}\times \{P_{n, \ell}(B)\}$.

\subsection{The Structure of the Optimal Feedback-Control Policy}
{
The stochastic optimization problems for feedback control with the discrete state space $\hat{\mathcal{X}}$ are formulated as follows. Define the corresponding feedback-control policy as $\hat{\mathcal{P}}:\hat{\mathcal{X}} \rightarrow\mathds{B}$. The matching average cost function  $\hat{\mathcal{L}}$ is modified from  \eqref{Eq:Lagrangian} as
\begin{equation}
\hat{\mathcal{L}}(\hat{\mathcal{P}}) =  \lim_{T\rightarrow\infty}\frac{1}{T}\E\l[\sum_{t=1}^T G(\hat{x}_t, B_t) \r]
\end{equation}
where $G(\hat{x}_t, B_t) $ is the cost-per-stage obtained using \eqref{Eq:TxCSI} as 
\begin{equation}\label{Eq:CostPerStage}
G(\hat{x}_t, B_t)  = \l\{
\begin{aligned}
&\hat{g}_t \E\l[\epsilon_t \mid B_t \r] + \lambda B_t, && B_t \geq 0 \\
&\hat{g}_t \hat{\delta}_t, && B_t = 0.
\end{aligned}
\r.
\end{equation}
Note that the minimum cost $\hat{\mathcal{L}}^\star$ converges to $\min_{\mathcal{P}}\mathcal{L}(\mathcal{P})$ as $N, M \rightarrow \infty$ provided that the grid points are suitably chosen \cite{Bertseka:ConvergeDiscretDynamicProg:75, HuangLau:EventDrivenFeedbackControlBeamforming}. 
The optimal policy $\hat{\mathcal{P}}^\star$  can be computed efficiently using \emph{policy iteration} \cite{Bertsekas07:DynamicProg}. The analysis of $\hat{\mathcal{P}}^\star$ is made tractable by considering a discounted-cost problem. Specifically,  given a discount factor $\rho \in (0, 1)$ and the initial state $\hat{x}_0$,   a stationary  feedback-control policy $\hat{\mathcal{P}}_\rho: \hat{\mathcal{X}}\rightarrow \mathds{B}$ is designed by minimizing the discounted cost function 
\begin{equation} \label{Eq:DistCost}
\mathcal{V}_\rho(\hat{\mathcal{P}}_\rho, \hat{x}_0) =  \sum_{t=0}^{\infty} \rho^t G(\hat{x}_t, B_t). 
\end{equation}
The optimal policy $\hat{\mathcal{P}}_\rho^\star$ and minimum cost $\mathcal{V}^\star_\rho$  converge to their average-cost counterparts as: $\hat{\mathcal{P}}^\star = \lim_{\rho\rightarrow 1}\hat{\mathcal{P}}^\star_\rho$  and $\hat{\mathcal{L}}^\star = \lim_{\rho\rightarrow 1} (1-\rho) \mathcal{V}^\star_\rho(\hat{x}_0)$ for arbitrary $\hat{x}_0$ \cite{Bertsekas07:DynamicProg}.  }

The discounted-cost problem allows simpler analysis as $\mathcal{V}_\rho^\star$ satisfies the following Bellman's equation: 
\begin{equation}\label{Eq:Bellman}
\mathcal{V}^\star_\rho(\hat{x}_t) = \mathsf{F}\mathcal{V}^\star_\rho(\hat{x}_t), \qquad \forall \ \hat{x}_t
\end{equation}
where $\mathsf{F}$ is the dynamic-programming operator and defined for a given function $q:\hat{\mathcal{X}}\rightarrow \mathds{R}$ as 
\begin{equation}\label{Eq:DPOperator}
\mathsf{F} q (\hat{x}_t) = \min_{B\in\mathds{B}}\l\{G(\hat{x}_t, B) + \rho\E\l[q(\hat{x}_{t+1})\mid \hat{x}_t, B\r]\r\}. 
\end{equation}
Though solving Bellman's equation analytically is infeasible, we can derive  from this equation some properties   of the optimal policy as follows. Several auxiliary results are obtained as shown in the following two lemmas.  First, the monotonicity  of $\mathcal{V}^\star_\rho(\hat{x})$ depends on if the following function is negative or nonnegative
\begin{equation}
\begin{aligned}
f(\bar{g}_k, \bar{\delta}_\ell) =  \mathcal{V}^\star_\rho(\bar{g}_k, \bar{\delta}_\ell)& - \mathcal{V}^\star_\rho(\bar{g}_k, \bar{\delta}_{\ell-1}) - \\
& \mathcal{V}^\star_\rho(\bar{g}_{k-1}, \bar{\delta}_\ell)+ \mathcal{V}^\star_\rho(\bar{g}_{k-1}, \bar{\delta}_{\ell-1}). 
\end{aligned}\label{Eq:f:Fun}
\end{equation}
with $\mathcal{V}^\star_\rho(\bar{g}_{k}, \bar{\delta}_\ell)=0$ if either $k=0$ or $\ell =0$. 
\begin{lemma}\label{Lem:Fun:f} \emph{The function $f(\hat{x})$ is nonnegative for all $\hat{x} \in \hat{\mathcal{X}}$. }
\end{lemma}
\begin{proof}The proof uses  the \emph{value iteration}, namely that for an arbitrary function $q:\hat{\mathcal{X}}\rightarrow \mathds{R}$,  the minimum discounted cost is \cite{Bertsekas07:DynamicProg}
\begin{equation}\label{Eq:ValIterate}
\mathcal{V}^\star_\rho(\hat{x}_t) = \lim_{n\rightarrow\infty} \mathsf{F}^n q (\hat{x}_t). 
\end{equation}
We show that if  $q$ is chosen to have the  property in the lemma statement, this property also holds for $\mathsf{F} q$ or in other words, remains unchanged by the dynamic-programming operation. Combining this fact and the value iteration in \eqref{Eq:ValIterate} proves the lemma. The details are provided in Appendix~\ref{App:Fun:f}. \end{proof}
Lemma~\ref{Lem:Fun:f} shows that $f(\hat{g}, \hat{\delta})$ is a monotone increasing function of $(\hat{g}, \hat{\delta})\in\hat{\mathcal{X}}$. Next, define the function
\begin{equation}\label{Eq:Z:Fun}
Z(\hat{x}_t, B) = G(\hat{x}_t, B) + \rho\E\l[\mathcal{V}^\star_\rho(\hat{x}_{t+1})\mid \hat{x}_t, B\r]. 
\end{equation}
Given the relation 
\begin{equation}\label{Eq:Relate:PZ}
\hat{\mathcal{P}}^\star_\rho(\hat{x}_t) = \arg\min_B Z(\hat{x}_t, B), 
\end{equation}
the structure of $\hat{\mathcal{P}}^\star_\rho$ depends directly on the characteristics of $Z$, which are specified in the following lemma. 
\begin{lemma} \label{Lem:JBFunc} \emph{$Z(\hat{x}_t, B)$ has the following properties. 
\begin{enumerate}
\item With $\hat{x}_t$ fixed  and for $B\in\mathds{B}$ and $B\neq 0$, $Z(\hat{x}_t, B)$ is a monotone decreasing and convex function of $B$;
\item With  $B$ fixed, $Z(\hat{g}_t, \hat{\delta}_t, B)$ is a monotone increasing function of $\hat{g}_t$ and also of $\hat{\delta}_t$ if $B =0$,  and of $\hat{g}_t$ and independent with $\hat{\delta}_t$ if $B > 0$.  
\end{enumerate}}
\end{lemma} 
The proof is provided in  Appendix~\ref{App:JBFunc}. 
 Using Lemma~\ref{Lem:Fun:f} and \ref{Lem:JBFunc}, the key result of this section is obtained as follows. 
 
\begin{theorem} \label{Theo:Policy} \emph{The optimal feedback-control policy $\hat{\mathcal{P}}^\star$ has the following properties. 
\begin{enumerate}
\item If there exists $(a, b)\in \hat{\mathcal{X}}$ such that $\hat{\mathcal{P}}^\star(a, b) = 0$, $\hat{\mathcal{P}}^\star(a, \hat{\delta}) = 0$ for all $\hat{\delta} \in \hat{\mathcal{D}}$ and $\hat{\delta} \leq b$. 
\item If there exists $(a, b)\in \hat{\mathcal{X}}$ such that $\hat{\mathcal{P}}^\star(a, b) > 0$, $\hat{\mathcal{P}}^\star(a, \hat{\delta})=\hat{\mathcal{P}}(a, b) $  for all $\hat{\delta} \in \hat{\mathcal{D}}$ and $\hat{\delta} >  b$. 
\item If there exist $(a, b), (c, b) \in \hat{\mathcal{X}}$ such that $\hat{\mathcal{P}}^\star(a, b) > 0$ and $\hat{\mathcal{P}}^\star(c, b) > 0$, $\hat{\mathcal{P}}^\star(c, b) \geq \hat{\mathcal{P}}^\star(a, b) $  if $c \geq a$ and vice versa. 
\end{enumerate}}
\end{theorem}
The proof is presented in Appendix~\ref{App:Policy}. 
The structure  of  $\hat{\mathcal{P}}^\star$ as specified  in Theorem~\ref{Theo:Policy} is  illustrated in Fig.~\ref{Fig:Policy:Struct}, from which $\hat{\mathcal{P}}^\star$ is observed to be  opportunistic in nature. CSI feedback over a particular feedback link  is performed only when the corresponding CSIT error and/or interference channel gain are large. As a result, the optimal policy partitions the state space into the \emph{feedback} and \emph{no-feedback} regions similar to the on/off-feedback policy in \cite{HuangLau:EventDrivenFeedbackControlBeamforming}. The current policy that supports variable-rate feedback further partitions the feedback region into smaller regions and assigns them different numbers of feedback bits. Upon feedback, the number of feedback bits increases with the interference-channel gain. The  CSIT error observed prior to feedback affects the decision on if feedback should be performed but has no influence on the number of feedback bits upon feedback.  The reason is that the CSIT error after feedback is equal to the quantization error that is independent of CSIT error prior to feedback. 

\begin{figure}
\begin{center}
\includegraphics[width=7cm]{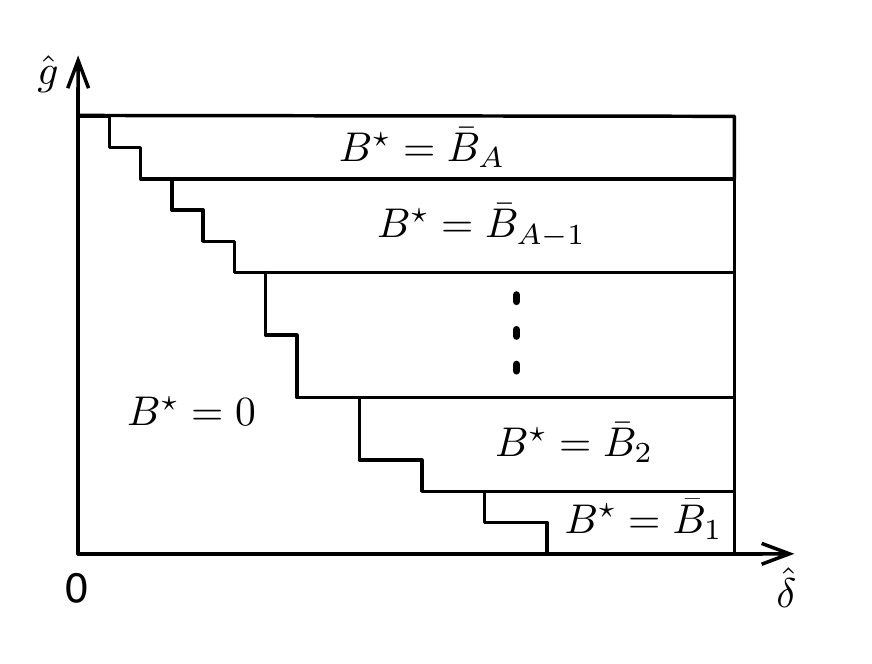}
\caption{The structure of the optimal feedback-control policy $\hat{\mathcal{P}}^\star$  where $0 \leq \bar{B}_1 \leq \cdots \leq \bar{B}_{A-1} \leq \bar{B}_A$ and  $\{\bar{B}_k\} \in \mathds{B}$.}  
\label{Fig:Policy:Struct}
\end{center}
\end{figure}

Intuitively, the feedback-link should be turned off less frequently when the interference-channel gain is large. In other words, the feedback-threshold function separating  the feedback and no-feedback regions should map larger values of $\hat{g}$ to smaller ones of  $\hat{\delta}$. However, proving this property requires more restrictive  assumptions on the channel temporal correlation than the current ones. 

The feedback control can be treated as the dual of bit loading (or adaptive modulation) over forward data links \cite{Goldsmith97, Chow95}. Both feedback control  and bit loading opportunistically allocate (CSI or data) bits over (feedback or forward) channels based on instantaneous (interference or data) CSI. Furthermore, both functions share the same objective of enhancing the system throughput. 

{ In wireless communication networks such as 3GPP-LTE, users are assigned dedicated (orthogonalized) feedback links. This approach incurs fast growing network overhead with the increasing popularity of cooperative transmission techniques such as multi-cell joint transmission \cite{Gesbert:MultiCellMIMOCooperativeNetworks:2010} or interference alignment \cite{CadJafar:InterfAlignment:2007}, for which the number of feedback links may  increase quadratically with the number of users.  Perhaps a more efficient approach is to allow multiple users to share a single feedback channel using e.g., a random access protocol. For this case, intelligent feedback control by receivers will alleviate feedback-traffic congestion and reduce the feedback delay for  CSI that is time sensitive. }

{
 
\subsection{The Computation of the Optimal Feedback-Control Policy} 
The  optimal feedback-control policy can be efficiently computed by  \emph{policy iteration} (see e.g.,  \cite{Bertsekas07:DynamicProg}). Each iteration involves  \emph{policy evaluation} and \emph{policy improvement}. The step of policy evaluation in the $i$-th iteration is to  compute  the corresponding average reward $\hat{\mathcal{L}}^{(i)}$ conditioned on a given policy $\hat{\mathcal{P}}^{(i-1)}$:
\begin{align}
\hat{\mathcal{L}}^{(i)} + \mathcal{U}^{(i)}_{m,n} &= G(\bar{g}_m, \bar{\delta}_n, \hat{\mathcal{P}}^{(i-1)}(\bar{g}_m, \bar{\delta}_n)) +\sum_{k,\ell}\mathcal{U}^{(i)}_{k,\ell}\tilde{P}_{m, k }\times\nn\\ 
&\qquad \qquad P_{n, \ell}(\hat{\mathcal{P}}^{(i-1)}(\bar{g}_m, \bar{\delta}_n)),\quad \forall \ m,n \label{Eq:PolIt:a}\\
\mathcal{U}^{(k)}_{1, 1} &= 0 \label{Eq:PolIt:b}
\end{align}
where $\l\{\mathcal{U}^{(i)}_{m,n}\r\}$ represents a set of scalars called \emph{differential rewards} and the constraint in \eqref{Eq:PolIt:b} ensures that the solution of \eqref{Eq:PolIt:a} is unique. In the ensuing  step of policy improvement, a new policy $\hat{\mathcal{P}}^{(i)}$ is computed using $\hat{\mathcal{L}}^{(i)}$ and $\l\{\mathcal{U}^{(i)}_{m,n}\r\}$ obtained by solving  \eqref{Eq:PolIt:a} and \eqref{Eq:PolIt:b}: 
\begin{equation}
\begin{aligned}
\hat{\mathcal{P}}^{(i)}(\bar{g}_m, \bar{\delta}_n) = \arg\max_{B\in\mathds{B}}& \biggl[G(\bar{g}_m, \bar{\delta}_n, B) +\\
& \sum_{k,\ell}\mathcal{U}^{(i)}_{k,\ell}\tilde{P}_{m, k}P_{n, \ell}(B)\biggr]
\end{aligned}
\end{equation}
for all $m$ and $n$. 
The above two steps are repeated till  the policy converges, namely $\hat{\mathcal{P}}^{(i+1)} = \hat{\mathcal{P}}^{(i)}$, yielding the optimal feedback-control policy.  As observed from simulation, the policy iteration converges typically within several iterations.  
}

{
\section{The Optimal Feedback-Control Policy: High Mobility}\label{Section:FBControl:HiMob}
In this section, we focus on the regime of high mobility  and derive more elaborate structural results  for the optimal feedback-control  policy  by directly solving the optimization problem \eqref{Eq:OpProb:AvFb} rather than relying on dynamic programing.

\subsection{The Structure of the Optimal Feedback-Control Policy}}
To simplify the solution of \eqref{Eq:OpProb:AvFb}, we consider the sphere-cap-quantized-CSI model in Example~\ref{Ex:SphereCap}, resulting in the optimal feedback-control policy  of the water-filling type as shown in the sequel.  This property   is expected to also hold for the  random-vector quantization in Example~\ref{Ex:RVQ} since the quantization-error expectations for both models have similar exponential forms (compare  \eqref{Eq:ExpErr:SCap} and \eqref{Eq:ExpErr:RVQ}).

Given independent block fading and a stationary feedback-control policy, the optimal feedback decisions in different slots are made independently. Consequently, $(g_t, \delta_t, B_t)$ have stationary distributions and are i.i.d. in different slots. To simplify notation, let $(g, \delta, B)$ represent a sample of  $\l\{g_t, \delta_t, B_t\r\}$ in an arbitrary slot. Using this notation and \eqref{Eq:ExpErr:SCap},  \eqref{Eq:OpProb:AvFb} can be rewritten as follows: 
\begin{equation} \label{Eq:OpProb:AvFb:min}
\begin{aligned}
&\underset{B}{\textrm{minimize:}}&& \E\l[g \min\l(\frac{L-1}{L}2^{-\frac{B}{L-1}}, \delta\r)\r]\\
&\textrm{subject to :}&& \E\l[B\r]\leq \frac{\bar{b}}{K-1}\\
&&&  B \in \mathds{B}
\end{aligned}
\end{equation}
where the $\min$ operator in the objective function accounts for the fact that 
feedback from a receiver to a particular interferer should be performed only if it reduces the expected CSIT error.   
Solving the problem in \eqref{Eq:OpProb:AvFb:min} analytically is difficult due to the constraint $B\in \mathds{B}$. To overcome this difficulty, the constraint $B\in \mathds{B}$ is relaxed as   $B\geq 0$ which approximates the case where many quantization resolutions are supportable.  The above optimization problem is modified accordingly as: 
\begin{equation} \label{Eq:OpProb:AvFb:min:a}
\begin{aligned}
&\underset{B}{\textrm{minimize:}}&& \E\l[g \min\l(\frac{L-1}{L}2^{-\frac{B}{L-1}}, \delta\r)\r]\\
&\textrm{subject to :}&& \E\l[B\r]\leq \frac{\bar{b}}{K-1}\\
&&&  B \geq 0. 
\end{aligned}
\end{equation}
Solving the above problem yields  the structure of the optimal feedback-control policy as described in the following proposition. 
\begin{proposition}\label{Prop:WaterFill:AvFb}\emph{
For high mobility, the optimal feedback-control  policy $\mathcal{P}^\star: \mathcal{X}\rightarrow \mathds{R}^+$ resulting from solving \eqref{Eq:OpProb:AvFb:min:a}
 is of the water-filling type:
\begin{equation}\label{Eq:WaterFill:AvFb}
\mathcal{P}^\star(g, \delta) =\l\{\begin{aligned} &\Upsilon - (L-1)\log_2\frac{1}{g}, && \delta \geq \Psi(g)\\
&0, && \textrm{otherwise}
\end{aligned}\r.
\end{equation}
where  $\Upsilon$ is the water level given  as 
\begin{equation}
\Upsilon = \frac{\bar{b}}{(K-1)\Pr(\delta \geq \Psi(g))} + (L-1)\E\l[\log_2\frac{1}{g}\mid \delta \geq \Psi(g)\r]. \nn
\end{equation}
The feedback-threshold function $\Psi: \mathcal{G}\rightarrow\mathcal{D}$ solves the following optimization problem: \footnote{The operator  $(a)^+$ for $ a\in\mathds{R}$ gives $a$ if $a \geq 0$ or otherwise $0$. }
\begin{equation} \label{Eq:OpProb:AvFb:c}
\begin{aligned}
\textrm{minimize:} &\quad \bar{I}^\star(\Psi)&\\
\textrm{subject to:}&\quad \Upsilon(\Psi) - (L-1)\log_2\frac{1}{\Psi^{-1}(1)}\geq\\
& \qquad \qquad \l((L-1)\log_2\frac{L-1}{L\Psi^{-1}(1)}\r)^+  
\end{aligned}
\end{equation}
where $\bar{I}^\star(\Psi)$  given below is the sum-interference power at any receiver achieved by $\mathcal{P}^\star$ given $\Psi$
\begin{equation}\label{Eq:IPwr:Min}
\begin{aligned}
\bar{I}^\star =\frac{1}{L-1} 2^{-\frac{\Upsilon}{L-1}}\Pr(\delta &\geq \Psi(g)) + \\
&\frac{1}{L-1}\E\l[g\delta \mid \delta < \Psi(g)\r]. 
\end{aligned}
\end{equation}
In addition, $\Psi(g)$  is a monotone decreasing function of $g$. }
\end{proposition}
The proof is provided in  Appendix~\ref{App:WaterFill:AvFb}. 
The above policy structure is consistent with that of the general solution as 
described by  Theorem~\ref{Theo:Policy} and its remarks. Moreover, the optimal feedback-control for high mobility is similar to the classic adaptive modulation algorithm  that  allocates data bits in time also based on water-filling \cite{Goldsmith97}. 

For a large average feedback rate $\bar{b} \gg 1$, $\Pr(\delta \geq \Psi(g))\approx 1$ and thus the minimum  average sum-interference power at an arbitrary receiver follows from \eqref{Eq:IPwr:Min} as
\begin{align}
\bar{I}^\star &\approx \frac{1}{L-1} 2^{-\frac{\Upsilon}{L-1}}\nn\\
&= c 2^{-\frac{\bar{b}}{(K-1)(L-1)}}\label{Eq:iPower:Min}
\end{align}
where $c$ is a constant. 
It can be observed from \eqref{Eq:iPower:Min} that $\bar{I}^\star$ decreases exponentially with increasing $\bar{b}$, where the slope is smaller for a larger number of links or  transmit antennas per transmitter.  In addition, the optimal number of feedback bits given in \eqref{Eq:WaterFill:AvFb} needs to be rounded to the nearest and smaller integer for implementation and this operation increases $\bar{I}^\star$ by a multiplicative  factor no larger than $2^{-\frac{1}{(K-1)(L-1)}}$. 

It is infeasible to obtain the feedback-threshold function  $\Psi$ analytically by solving the optimization problem in Proposition~\ref{Prop:WaterFill:AvFb}. Thus  computing $\Psi$ requires   a numerical search, which is used to obtain relevant simulation results in Section~\ref{Section:Simulation}.

{
Finally, we obtain some insight into the effect of quantizing  direct-link feedback  (feedback from a receiver to the intended transmitter) and justify its omission in the performance metric. Consider an arbitrary data link in the current MISO interference channel, where the transmit beamformer, channel-direction vector, and received interference power  are denoted as $\bff_0$, $\bs_0$ and $I_0$, respectively. The direct-link feedback  of the quantized version $\hat{\bs}_0$ of $\bs_0$ allows the transmitter to perform the maximum-ratio transmission under the constraint of zero-forcing beamforming \cite{Jindal:RethinkMIMONetwork:LinearThroughput:2008}. As a result, the corresponding effective channel gain after beamforming can be shown to be   $\phi(1-\zeta)$ where $\phi$ follows the chi-square distribution with $2(L-K+1)$ degrees of freedom and $\zeta$ is no larger than  the quantization error $\epsilon_0$ of $\bs_0$, where $\epsilon_0 = 1- |\hat{\bs}_0^\dagger\bs_0|^2$ \cite{Jindal:RethinkMIMONetwork:LinearThroughput:2008}. For a high SINR and small $\zeta$, throughput $R$ of the considered data link can be approximated as follows
\begin{eqnarray}
R &\approx& \E\l[\log_2\frac{\phi(1-\zeta)}{I_0}\r]\nn\\
&\approx& -\E[\zeta] -\E[\log_2I_0] + \E[\log_2 \phi]\nn\\
&\geq& -\E[\epsilon_0]  - \log_2\E[I_0] + \E[\log_2 \phi]. \label{Eq:LocalFeedback}
\end{eqnarray}
Assuming that $\hat{\bs}_0$ is generated  by a random vector quantizer, it follows from  \eqref{Eq:ExpErr:RVQ} that $\E[\epsilon_0] \approx a e^{-\frac{1}{L-1}\check{B}_0}$ where $\check{B}_0$ denotes the number of direct-link-feedback bits. Moreover, for high mobility and a large cooperative feedback rate, $\E[I_0]$ can be approximated by $\bar{I}^\star $ given in \eqref{Eq:iPower:Min}. Then it follows from \eqref{Eq:iPower:Min} and  \eqref{Eq:LocalFeedback} that 
\begin{eqnarray}
R &\geq& -a e^{-\frac{1}{L-1}\check{B}_0}  + \frac{\bar{b}}{(K-1)(L-1)} + \text{constant}.\label{Eq:LocalFeedback:a}
\end{eqnarray}
The first term at the right-hand side of \eqref{Eq:LocalFeedback:a} represents the throughput loss due to the direct-link-feedback error  and the second the throughput gain obtained by increasing the cooperative feedback rate. It can be observed that the effect of the direct-link-feedback error  diminishes exponentially with $\check{B}_0$ and hence omitted in the current analysis.  }

{
\subsection{Extension to Asymmetric Channel Distributions}\label{Section:AsymChan}}

In the preceding sections, all interference channels are assumed to follow identical distributions. In this section, we discuss feedback control for asymmetric interference channel distributions in terms of heterogeneous path losses and assuming high mobility for mathematical tractability. Let $d^{[mn]}$ denote the distance between receiver $m$ and transmitter $n$.  The average interference power at receiver $m$ can be written as
\begin{equation}
\bar{I}^{[m]} = \sum_{n \neq m } (d^{[mn]})^{-\alpha} \bar{I}^{[mn]}  
\end{equation}
where $\alpha$ is the path-loss exponent and 
\begin{equation}
\bar{I}^{[mn]}  
= \frac{1}{L-1}\E\l[ g_t^{[mn]} \min\l(\E\l[\epsilon^{[mn]}\mid B^{[mn]}_t\r], \delta_t^{[mn]}\r)\r]. \label{Eq:AvI:m}
\end{equation}

Given heterogeneous path losses, the uniform allocation of average feedback rates by each receiver  to different feedback channels is no longer optimal. Consequently, we should optimize the average feedback-rate allocation besides feedback-control over time. Specifically, the feedback-control optimization problem can be decomposed as: 
\begin{enumerate}

\item[--] Master problem (average feedback-rate allocation)
\begin{equation}\label{Eq:MasterProb}
 \begin{aligned}
&\underset{\{\bar{b}_{m, n}\}}{\textrm{minimize:}} &&  \sum_{\substack{n = 1\\ n \neq m }}^K \l(d^{[mn]}\r)^{-\alpha} \bar{I}^{[mn]}_{\min}\l(\bar{b}_{m, n}\r)\\
&\textrm{subject to:} &&  \sum_{\substack{n = 1\\ n \neq m }}^K \bar{b}_{m, n}\leq \bar{b}\\
&&&\bar{b}_{m, n} \geq 0 \ \forall \ m \neq n 
\end{aligned}
\end{equation}
where $\bar{I}^{[mn]}_{\min}\l(\bar{b}_{m, n}\r)$ solves the following sub-problem.

\item[--] Sub-problem (stochastic feedback control)
\begin{equation}\label{Eq:SubProb}
 \begin{aligned}
&\textrm{minimize:} &&  \bar{I}^{[mn]}(\mathcal{P}_{mn})\\
&\textrm{subject to:} && \E\l[B_t^{[mn]}\r]\leq \bar{b}_{m, n}
\end{aligned}
\end{equation}
where $\bar{I}^{[mn]}$ is given in \eqref{Eq:AvI:m} and $\mathcal{P}_{mn}$ denotes the stationary policy for controlling the feedback link from receiver $m$ to transmitter $n$. 
\end{enumerate}
Note that the sub-problem is identical to  \eqref{Eq:OpProb:AvFb} except for the difference in the maximum average feedback rates.  The above decomposed optimization problems have an unique solution as shown below. 
\begin{lemma}\label{Lem:Convex}\emph{$\bar{I}^{[mn]}_{\min}\l(\bar{b}_{m, n}\r)$ is a convex and monotone decreasing function over $\bar{b}_{m, n} \geq 0$. }
\end{lemma}
The proof is presented in  Appendix~\ref{App:AsymChan}.   The following result holds given  the convexity of the master problem as a result of Lemma~\ref{Lem:Convex} and that of the sub-problem follows from the discussion  in Section~\ref{Section:FBControl:HiMob}. 
\begin{proposition}\label{Prop:AsymChan} \emph{Solving the master problem and sub-problem gives an unique optimal stationary feedback-control policy. }
\end{proposition}

Next, we characterize the optimal feedback-control policy based on the quantizer model in Example~$1$ and for a large average sum-feedback rate per user.  For this case, using \eqref{Eq:iPower:Min} and given $\bar{b}_{m, n}$, the average interference power from transmitter $n$ to receiver $m$ can be approximated as 
\begin{eqnarray}
\bar{I}^{[mn]} &\approx& \frac{1}{L-1} 2^{-\E\l[\log_2\frac{1}{g}\r]} 2^{-\frac{\bar{b}_{m, n}}{L-1}}\end{eqnarray}
where $g$ follows the chi-square distribution. 
This approximation reduces the master problem as:
\begin{equation}
 \begin{aligned}
&\underset{\{\bar{b}_{m, n}\}}{\textrm{minimize:}} &&  \sum_{\substack{n = 1\\ n \neq m }}^K \l(d^{[mn]}\r)^{-\alpha} 2^{-\frac{\bar{b}_{m, n}}{L-1}}\\
&\textrm{subject to:} &&  \sum_{\substack{n = 1\\ n \neq m }}^K \bar{b}_{m, n}\leq \bar{b}\\
&&&\bar{b}_{m, n} \geq 0 \ \forall \ m \neq n. 
\end{aligned}\nn
\end{equation}
Solving the above constrained optimization problem using Lagrangian method yields that the optimal allocation of average feedback rates is of the water-filling type:
\begin{equation}\label{Eq:RateSplit}
\bar{b}^\star_{m,n} = \eta - \alpha (L-1)\log_2 d^{[mn]}, \quad n \neq m 
\end{equation}
where $\eta$ is the water-level given as
\begin{equation}
\eta = \frac{\bar{b}}{K-1} + \frac{\alpha(L-1)}{K-1}\sum_{n\neq m}\log_2d^{[mn]}. 
\end{equation}
For a sanity check, the substitution of equal distances $d^{[m1]} = d^{[m2]}=\cdots = d^{[mK]}$ into \eqref{Eq:RateSplit} gives equal-rate splitting: $\bar{b}^\star_{m,n} = \frac{\bar{b}}{K-1}$ for all $ n \neq m$. It can be observed from \eqref{Eq:RateSplit} that the optimal average feedback rate allocated by a receiver for suppressing the interference-power of a particular interferer decreases logarithmically with the increasing  distance between the interferer and the receiver. Relaxing the integer constraint on the numbers of feedback bits and combining \eqref{Eq:WaterFill:AvFb} and \eqref{Eq:RateSplit}, we can approximate the optimal number of feedback  bits $B_{m, n}^\star$ sent from receiver $m$ to transmitter $n$ with $ n\neq m$ as 
\begin{equation}\label{Eq:WaterFill:FB}
B_{m, n}^\star = \eta' - \alpha (L-1)\log_2 d^{[mn]} - (L-1)\log_2\frac{1}{g^{[mn]}} 
\end{equation}
where $\eta'$ is a constant. 
The above expression shows two-tier water-filling for allocating average feedback rates over multiple feedback links and for each link distributing feedback bits over different slots. 

{
The feedback scheme in \eqref{Eq:WaterFill:FB} is similar to those in \cite{RamyaHeath:AdaptiveBitPartionDelayedLimFb, KhoYu:LimFBCodebookDesignMultiuserCase:2010} in that the optimal number of feedback bits for a particular feedback link increases logarithmically with the channel gain of the corresponding forward link, despite the differences in settings (interference networks, cooperative multi-cell networks \cite{RamyaHeath:AdaptiveBitPartionDelayedLimFb}, or multiuser downlink systems \cite{KhoYu:LimFBCodebookDesignMultiuserCase:2010}) and metrics (sum interference power, throughput loss \cite{RamyaHeath:AdaptiveBitPartionDelayedLimFb} or total transmission power \cite{KhoYu:LimFBCodebookDesignMultiuserCase:2010}). The fundamental reason for the above similarity is that different performance optimization problems can be reduced to or approximated by one that minimizes a weighted sum of exponential functions of numbers of feedback bits under a constraint on the sum-feedback rate. }

\begin{figure}
\begin{center}
\hspace{-20pt}\subfigure[$f_d = 10^{-2}$ and $\bar{b} = 12$ bit/slot]{\includegraphics[width=8.5cm]{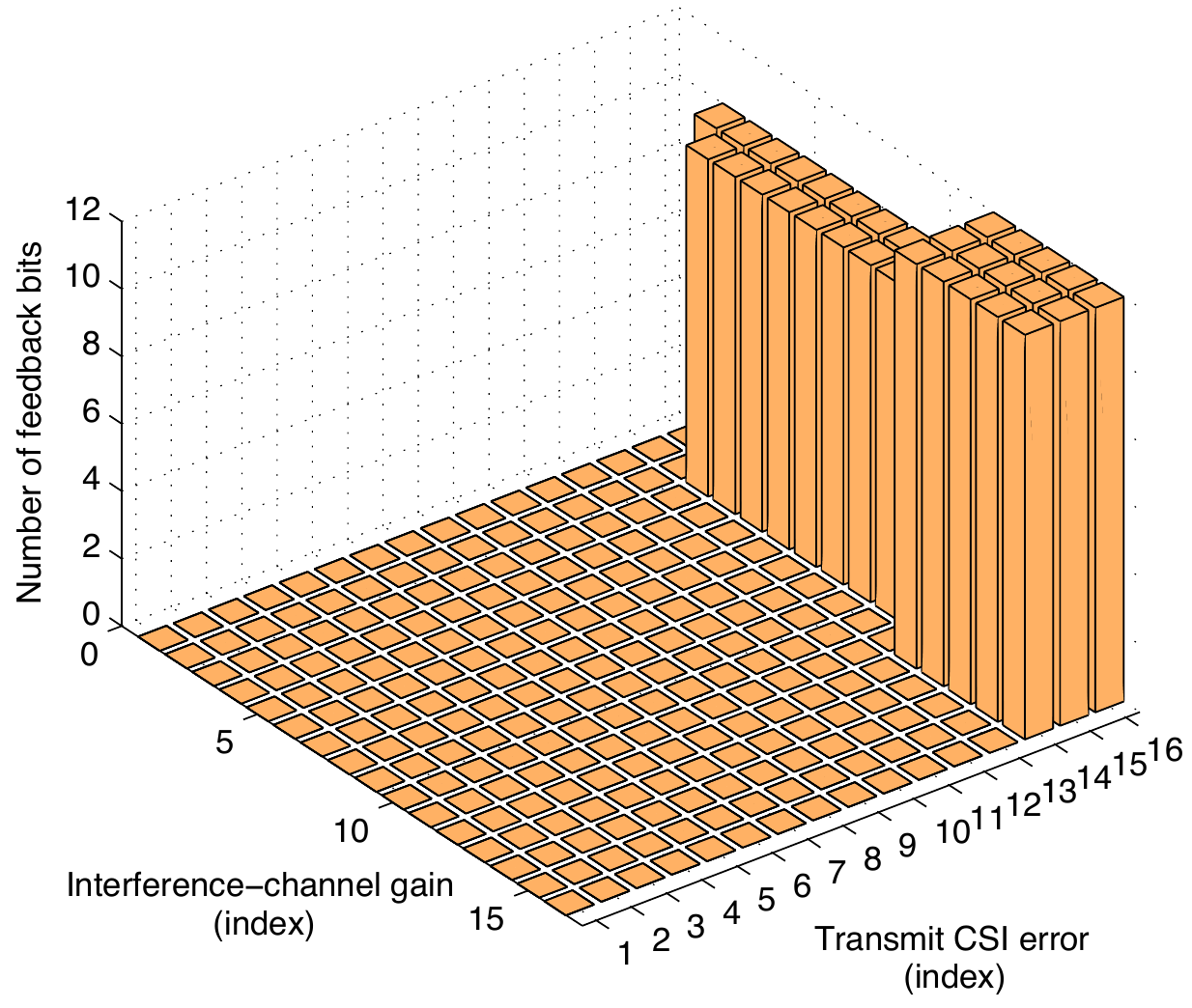}\hspace{-10pt}}
\subfigure[$f_d = 6\times10^{-3}$ and $\bar{b} = 36$ bit/slot]{\includegraphics[width=8.5cm]{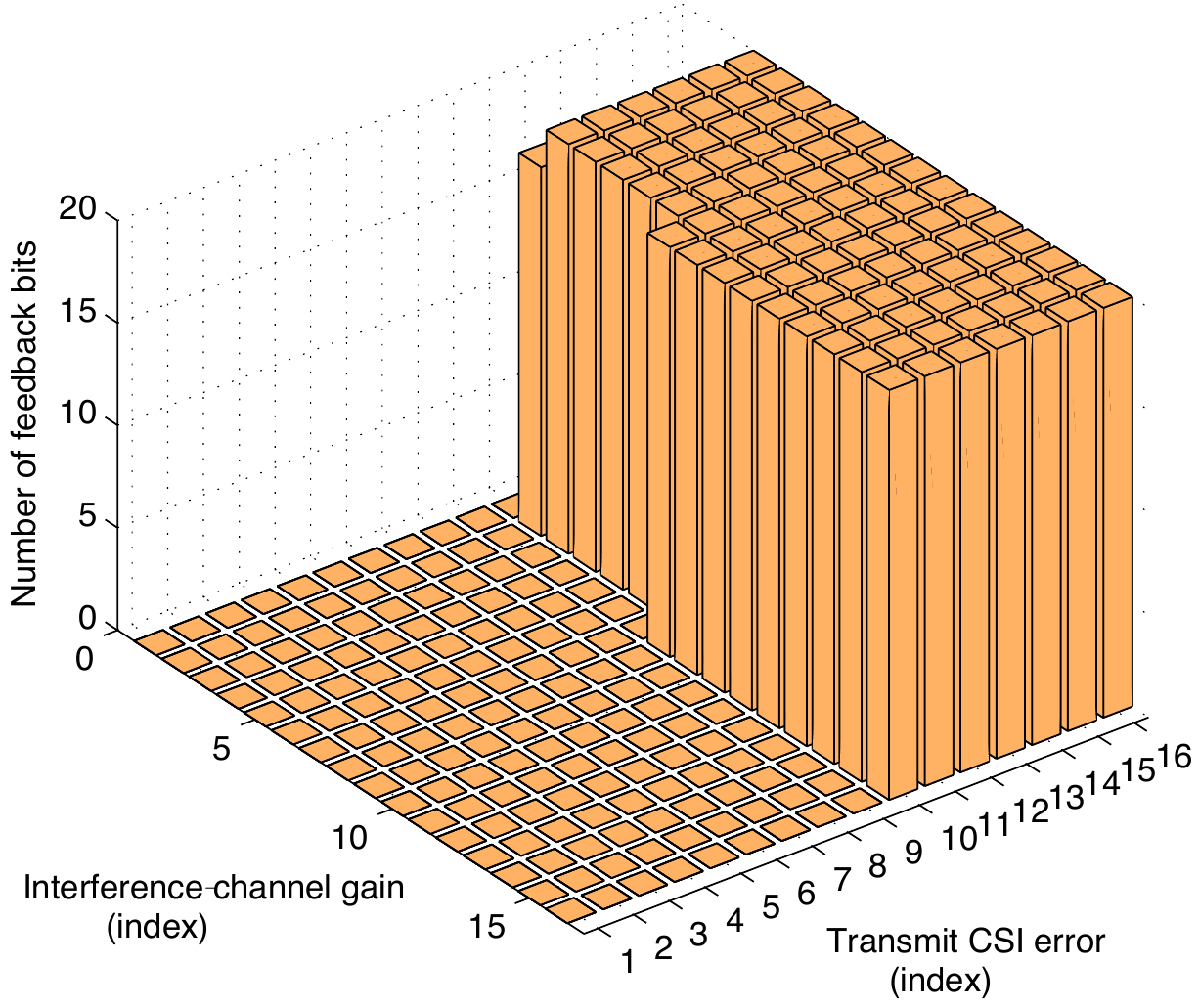}\hspace{-20pt}}
\caption{Optimal feedback-control policies  given average sum-feedback constraints and a discrete  state space}
\label{Fig:Policy}
\end{center}
\end{figure}

\section{Simulation Results}\label{Section:Simulation}
The simulation has the following settings unless specified otherwise.   The number of antennas $L=4$, the number of users $K=3$, and  the set of available numbers of feedback bits is $\mathds{B} = \{2n\mid 0\leq n \leq 15\}$. All channel fading coefficients are modeled as i.i.d. $\mathcal{CN}(0, 1)$ Gaussian processes. For low-to-moderate mobility,  the temporal correlation of each process is specified by Clark's function \cite{JakesBook}.  The values of Doppler frequency are normalized by the symbol rate.  
The state space for feedback control at low-to-moderate mobility is discretized to have $M=16$ grid points for the interference-channel gain  and $N=16$ points for the CSIT error. The set  $\hat{\mathcal{G}}$ is chosen based on the equal-probability criterion such that $\Pr( \tilde{g}_{k} \leq g^{[mn]}< \tilde{g}_{k+1})=\frac{1}{M}$ for $1\leq k \leq M$, $\tilde{g}_1 = 0$ and $\tilde{g}_{M+1} = \infty$. The CSI quantization error is generated based on the sphere-cap-quantized-CSI model in Example~\ref{Ex:SphereCap}. Correspondingly, the grid points for the CSIT error are chosen to be the expected quantization errors for different numbers of feedback bits in $\mathds{B}$, namely $\hat{\mathcal{D}} = \l\{\frac{L-1}{L}2^{-\frac{B}{L-1}}\mid B\in \mathds{B}\r\}$. 

\begin{figure}
\begin{center}
\includegraphics[width=8.7cm]{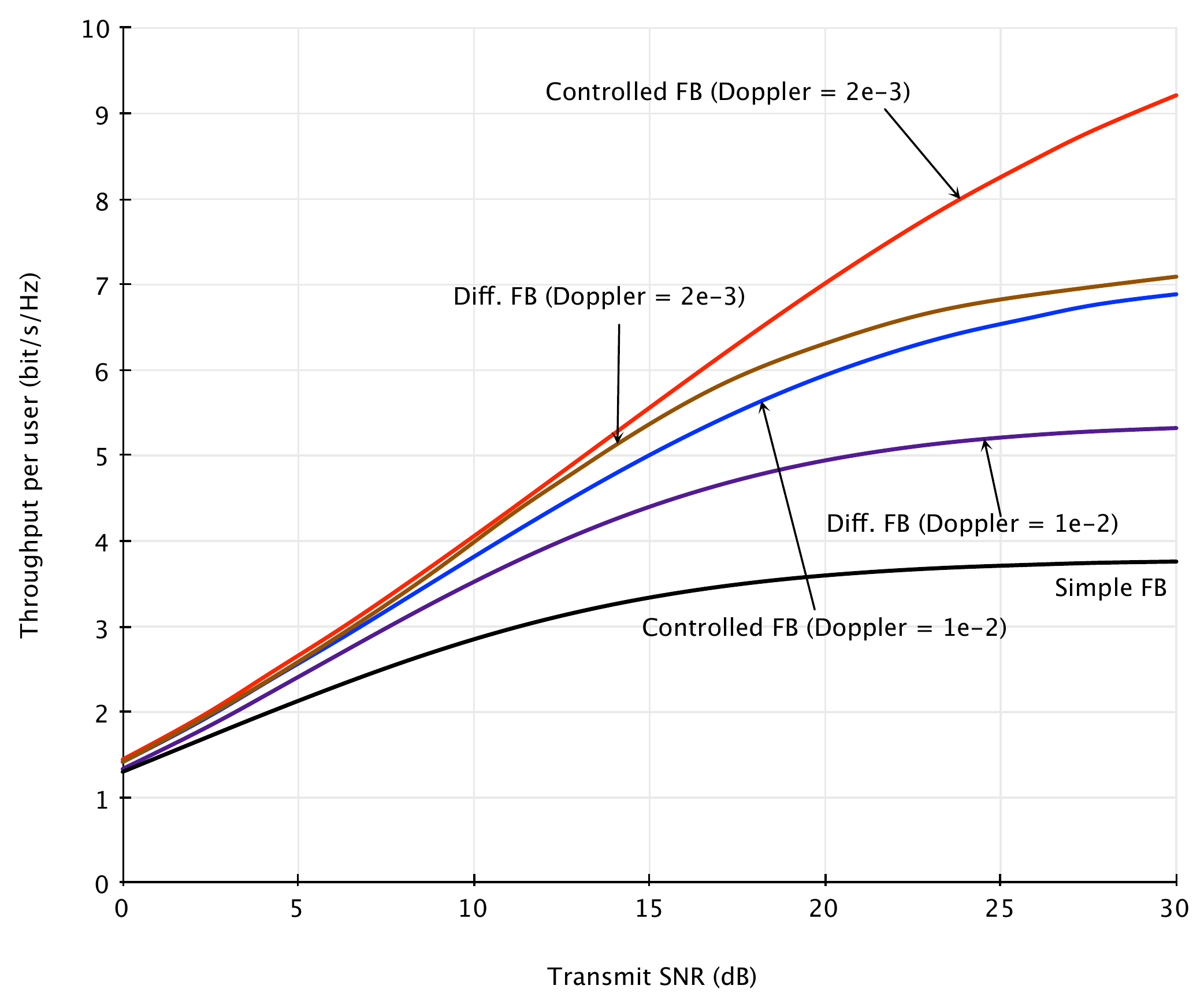}
\caption{Throughput-per-user versus transmit SNR for low-to-moderate mobility and different limited-feedback techniques  under an average sum-feedback constraint per user  of  $16$ bit/slot.  }
\label{Fig:CapacitySNR}
\end{center}
\end{figure}

Fig.~\ref{Fig:Policy} to \ref{Fig:CapacityFbRate}  concern stochastic feedback control for low-to-moderate mobility. Fig.~\ref{Fig:Policy} shows the optimal feedback-control policies computed using policy iteration for different combinations of (normalized) Doppler frequency $f_d$  and average sum-feedback rates $\bar{b}$. Both Fig.~\ref{Fig:Policy}(a) and \ref{Fig:Policy}(b) are consistent with Theorem~\ref{Theo:Policy}. Specifically, it can be observed from the figures that given the optimal policy, the state space is partitioned into the feedback and no feedback regions. Moreover, in the feedback region, $B^\star$ is independent of the CSIT error $\hat{\delta}$; given $\hat{\delta}$, $B^\star$ is a monotone non-decreasing function of $\hat{g}$. Comparing Fig.~\ref{Fig:Policy}(a) and \ref{Fig:Policy}(b), increasing Doppler frequency and the average sum-feedback rate enlarge the feedback region as well as the numbers of feedback bits in the feedback region. 

{
Fig.~\ref{Fig:CapacitySNR} shows  the throughput-per-user versus transmit SNR for optimally controlled feedback given $\bar{b} = 12$ bit/slot and for conventional feedback algorithms with a sum feedback constraint of $16$ bit/slot, where the additional $4$ bit/slot accounts for the extra feedback-control overhead for specifying a varying number of feedback bits. For comparison, two existing feedback methods are considered, namely \emph{simple feedback} for which  CSI in each slot is quantized with a fixed resolution ($8$ bits) and feedback is performed in each slot (see e.g., \cite{LovHeaETAL:GrasBeamMultMult:Oct:03}) and \emph{differential feedback} that exploits channel temporal correlation for feedback reduction (see e.g., \cite{KimLov:MIMODiffFBSlowFading:2011}). The different-feedback algorithm considered here is from \cite{KimLov:MIMODiffFBSlowFading:2011} and allows transmitter $n$ to construct the channel direction $\hat{\bs}_t^{[mn]}$ using the past CSI $\hat{\bs}_{t-1}^{[mn]}$ and a quantized $L\times L$ unitary matrix $\Lambda^{[mn]}_t$ sent by receiver $m$ as follows:
\begin{equation}
\hat{\bs}_t^{[mn]} = \l(\sqrt{1-\nu^2} + \nu \Lambda^{[mn]}_t\r) \hat{\bs}_{t-1}^{[mn]}
\end{equation}
where $0 < \nu < 1 $ is adapted to Doppler frequency by a numerical search using the criterion of maximum throughput and $\Lambda^{[mn]}_t$ is chosen from a $8$-bit random codebook of i.i.d. entries such that the CSI error is minimized. 
From Fig.~\ref{Fig:CapacitySNR}, the  throughput-per-user for both simple and differential feedback is observed to saturate as the transmit SNR increases and residual interference  becomes dominant over noise. The use of feedback control alleviates this performance degradation and increases  the throughput-per-user significantly especial at high SNRs. Moreover, the throughput-per-user given feedback control increases rapidly as the Doppler frequency decreases, corresponding to growing redundancy in CSI. Specifically, reducing $f_d$ from $1\times 10^{-2}$ to $2\times 10^{-3}$ increases the throughput-per-user by up to about $2$ bit/s/Hz.  

\begin{figure}
\begin{center}
\includegraphics[width=8.7cm]{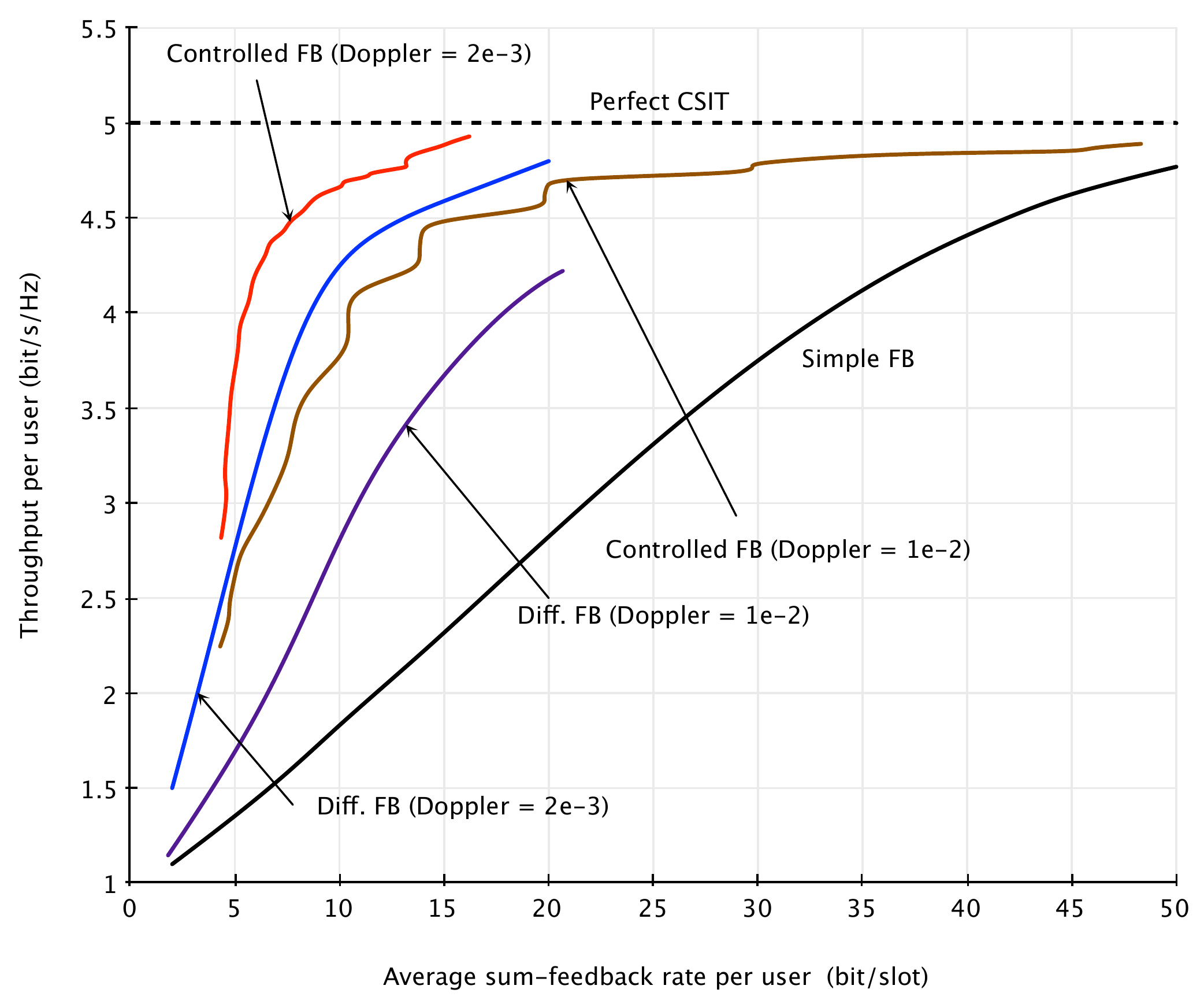}
\caption{Throughput-per-user versus average sum-feedback rate for different limited-feedback techniques with  low-to-moderate mobility and the transmit SNR equal to $13$ dB.}
\label{Fig:CapacityFbRate}
\end{center}
\end{figure}

Fig.~\ref{Fig:CapacityFbRate} shows  the throughput-per-user versus average sum-feedback rate per user for both controlled feedback and conventional feedback methods, where the additional controlled-feedback overhead mentioned earlier has been accounted for. It can be observed that as the average sum-feedback rate increases, the throughput-per-user for the optimally controlled feedback converges  to the upper bound corresponding to perfect CSIT faster than that for differential feedback and much more rapidly than that for simple feedback. Consequently, given the same average sum-feedback constraint, the optimal feedback control yields higher throughput than the two conventional methods. It can be observed that exploiting the channel temporal correlation by either feedback control or differential feedback can provide significant throughput gains. 
For example, feedback control increases the throughput-per-user of simple feedback by about $3$ times given the average sum-feedback rate of $7$ bit/slot and $f_d = 2\times 10^{-3}$. Last, note that the humps on the curves for controlled   feedback  are due to  discretization of the state space. }

Finally, we consider the optimal feedback control for high mobility. Fig.~\ref{Fig:CapacityIID} displays the curves of throughput-per-user  versus average sum-feedback rate  per user  for controlled feedback as well as no feedback control, namely that the  rates for different feedback links are equal and  simple feedback is applied.  These results are based on $\mathds{B} \in \mathds{N}^+$, aligned with the analysis in Section~\ref{Section:FBControl:HiMob}. It is observed that the throughput  gain of the optimal feedback control with respect to the case of no  feedback control is marginal given symmetric channel distributions and high mobility, namely no redundancy in CSI. However, this gain is significant in the presence of asymmetric channel distributions and  unequal distribution of average feedback rates over different feedback links.


\begin{figure}
\begin{center}
\includegraphics[width=8.7cm]{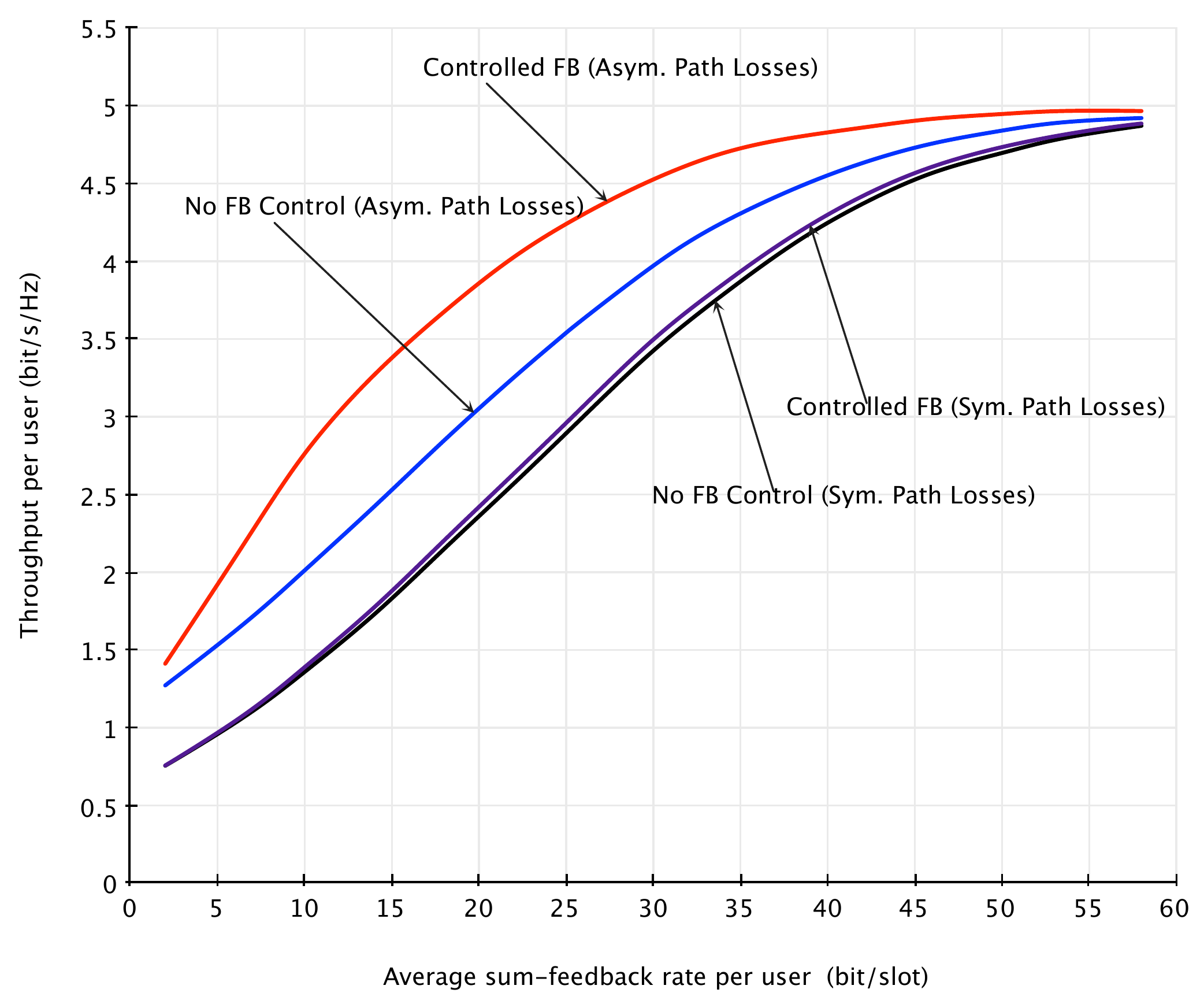}
\caption{Throughput-per-user versus average sum-feedback rate for  the optimal feedback control with  high mobility. All interference channels have unit propagation distances for the case of symmetric channel distributions. For the asymmetric case, the $(K-1) = 2$ interferers for each receiver are located at distances of $1$ and $3$ units away. The path-loss exponent is $\alpha = 3$ and transmit SNR $13$ dB. }
\label{Fig:CapacityIID}
\end{center}
\end{figure}

\section{Conclusion}\label{Section:Conclusion}
This work has proposed the new approach of  distributive and stochastic control of event-driven CSI feedback in multi-antenna interference networks.  For symmetric channel distributions, the optimal feedback-control policy for each feedback link  has been proved to be opportunistic. Specifically, feedback is performed only if the corresponding interference-channel gain is large or the CSI at the transmitter is significantly outdated; the number of feedback bits increases with the interference-channel gain. For high-mobility and symmetric channel distributions, by considering a specific CSI quantization model, the optimal feedback policy has been shown to be of the water-filling type that also  has the above opportunistic properties. For high-mobility and heterogeneous path-losses for the interference channels, the optimization of the feedback controller has been decomposed into a master problem and a sub-problem. We have proved the existence of an unique solution for the decomposed optimization problems. 

To the best of our knowledge, this is the first work on applying stochastic-optimization theory to design feedback controllers in multi-antenna interference networks. This work opens several issues for future investigation. First, in the case of bursty traffic, the queues and feedback-links can be jointly controlled to achieve the optimal tradeoff between transmission delay and feedback overhead. Second, the event-driven feedback targets shared feedback channels where feedback collisions are inevitable. Collisions and the resultant feedback delay are omitted in the current work but important issues to consider in designing practical feedback controllers and protocols. Last, it is challenging to generalize  the current  feedback-controller  designs to  more complex settings such as MIMO channels and spatial multiplexing,  and alternative beamforming algorithms such as one using the minimum-mean-square-error criterion.

\bibliographystyle{ieeetr}

\appendix

\subsection{Proof for Lemma~\ref{Lem:Fun:f}}\label{App:Fun:f}
For ease of notation, define a function $\Phi(\bar{g}_k\mid \bar{\delta}_b,  B^\star)$ as 
\begin{equation}\label{Eq:Phi:Fun}
\Phi(\bar{g}_k\!\mid\! \bar{\delta}_b,  B^\star) \!=\!\! \underset{\hat{\delta}_{t+1}}{\E}\!\!\l[\mathcal{V}_\rho^\star(\hat{g}_{t+1}\!=\!\bar{g}_k, \hat{\delta}_{t+1})\!\mid\! \hat{\delta}_{t}\!=\!\bar{\delta}_b,  B^\star\r].
\end{equation}
We can write that 
\begin{align}
\Phi(\bar{g}_k\mid \bar{\delta}_b,  B^\star) &= \sum_{n=1}^N \mathcal{V}_\rho^\star(\bar{g}_k, \bar{\delta}_n)P_{b, n}(B^\star)\nn\\ 
&= \sum_{n=1}^N\l[\mathcal{V}_\rho^\star(\bar{g}_k, \bar{\delta}_n) - \mathcal{V}_\rho^\star(\bar{g}_k, \bar{\delta}_{n-1})\r]\times\nn\\
&\qquad \sum_{m=n}^N P_{b, m}(B^\star) \nn
\end{align}
with $\mathcal{V}_\rho^\star(\bar{g}_k, \bar{\delta}_{0})=0$. 
Similarly, from \eqref{Eq:Phi:Fun}, we can obtain that 
\begin{align}
\E[&\mathcal{V}_\rho^\star(\hat{g}_{t+1}, \hat{\delta}_{t+1})\mid \hat{g}_t=\bar{g}_a, \hat{\delta}_t = \bar{\delta}_b, B^\star] \nn\\
& =\E[\Phi(\hat{g}_{t+1}\mid\bar{\delta}_b,  B^\star)\mid \hat{g}_{t}=\bar{g}_a]\nn\\
&=\sum_{k=1}^M \l[\Phi(\bar{g}_k\mid \bar{\delta}_b,  B^\star) - \Phi(\bar{g}_{k-1}\mid\bar{\delta}_b,  B^\star)\r] \sum_{\ell=k}^M \tilde{P}_{a, \ell}\nn\\
&= \sum_{n=1}^N \sum_{k=1}^M f(\bar{\delta}_n, \bar{g}_k)\sum_{m=n}^N P_{b, m}(B^\star)\sum_{\ell=k}^M \tilde{P}_{a, \ell}\label{Eq:J:Cond:Exp}
\end{align}
with $\Phi(\bar{g}_{0}\mid\bar{\delta}_b,  B^\star)=0$. 
Using  \eqref{Eq:CostPerStage}, \eqref{Eq:Bellman}, and \eqref{Eq:DPOperator}, it can be obtained for $B^\star = 0$ that 
\begin{align}
\mathsf{F}\mathcal{V}_\rho^\star(\bar{g}_a, \bar{\delta}_b) &= \bar{g}_a\bar{\delta}_b +  \E\l[\mathcal{V}_\rho^\star(\hat{g}_{t+1}, \hat{\delta}_{t+1})\mid \hat{g}_t=\bar{g}_a, \hat{\delta}_t = \bar{\delta}_b\r]\nn\\
&= \bar{g}_a\bar{\delta}_b +  \sum_{n=1}^N \sum_{k=1}^M f(\bar{\delta}_n, \bar{g}_k)\sum_{m=n}^N P_{b, m}(0)\sum_{\ell=k}^M \tilde{P}_{a, \ell} \label{Eq:J:B:1}
\end{align}
where \eqref{Eq:J:B:1} uses \eqref{Eq:J:Cond:Exp}. It follows from \eqref{Eq:f:Fun} and \eqref{Eq:J:B:1}  that
\begin{align}
\mathsf{F}f(\bar{g}_a, \bar{\delta}_b) 
&= \mathsf{F}\mathcal{V}^\star_\rho(\bar{g}_k, \bar{\delta}_\ell) - \mathsf{F}\mathcal{V}^\star_\rho(\bar{g}_k, \bar{\delta}_{\ell-1}) -\nn \\
& \qquad \mathsf{F}\mathcal{V}^\star_\rho(\bar{g}_{k-1}, \bar{\delta}_\ell)+ \mathsf{F}\mathcal{V}^\star_\rho(\bar{g}_{k-1}, \bar{\delta}_{\ell-1})\label{Eq:f:Fun:DP}\\
&= (\bar{g}_a - \bar{g}_{a-1})(\bar{\delta}_b -  \bar{\delta}_{b-1}) + \sum_{n=1}^N \sum_{k=1}^M f(\bar{\delta}_n, \bar{g}_k)\times\nn\\
&\qquad\l[\sum_{m=n}^N P_{b, m}(0)-\sum_{m=n}^N P_{b-1, m}(0)\r]\times\nn\\
&\qquad \qquad\l[\sum_{\ell=k}^M \tilde{P}_{a, \ell} - \sum_{\ell=k}^M \tilde{P}_{a-1, \ell}\r]\label{Eq:f:Fun:DP:a}
\end{align}
with $\bar{g}_0 = \bar{\delta}_0=0$. 
Given Assumption~\ref{AS:SD},  \eqref{Eq:f:Fun:DP:a} yields that $\mathsf{F} f(\hat{g},  \hat{\delta})\geq 0$ for all $(\hat{g}, \hat{\delta})$ if  $f(\hat{g},  \hat{\delta}) \geq 0$ for all $(\hat{g}, \hat{\delta})$ and $B^\star = 0$. Next, from \eqref{Eq:CostPerStage}, \eqref{Eq:Bellman}, and \eqref{Eq:DPOperator}, it can be obtained for $B^\star > 0$ that
\begin{equation}\label{Eq:J:B}
\begin{aligned}
\mathsf{F}\mathcal{V}_\rho^\star(\bar{g}_a, \bar{\delta}_b) = \bar{g}_a&\E[\epsilon\mid B^\star] + B^\star + \\
&\E\bigl[\mathcal{V}_\rho^\star(\hat{g}_{t+1}, \hat{\delta}_{t+1})\mid \hat{g}_t=\bar{g}_a, B^\star\bigr]. 
\end{aligned}
\end{equation}
It can be observed from \eqref{Eq:J:B} that $\mathsf{F}\mathcal{V}_\rho^\star(\bar{g}_a, \bar{\delta}_b)$ is independent with $\bar{\delta}_b$. Thus, using this fact and \eqref{Eq:f:Fun:DP} gives that 
$\mathsf{F} f(\hat{g}, \hat{\delta}) =0$ for $B^\star > 0$. By combining above results, we conclude that the policy iteration retains the property $f(\hat{g},  \hat{\delta}) \geq 0$ if its initialization has such a property (e.g., $f(\hat{g},  \hat{\delta}) = 1$ for all $(\hat{g},  \hat{\delta})$). This completes the proof.

\subsection{Proof for Lemma~\ref{Lem:JBFunc}}\label{App:JBFunc}
Using \eqref{Eq:CostPerStage}, \eqref{Eq:Z:Fun} and \eqref{Eq:J:Cond:Exp}, it can be obtained that 
 \begin{align}\label{Eq:J:B:1m}
Z&(a, b, B)  = \nn \\
&\l\{
\begin{aligned}
& a\E[\epsilon \mid B] + B + \sum_{n=1}^N \sum_{k=1}^M f(\bar{\delta}_n, \bar{g}_k)\\
&\qquad \sum_{m=n}^N P_{b, m}(B)\sum_{\ell=k}^M \tilde{P}_{a, \ell}, && B > 0\\
&ab + \sum_{n=1}^N \sum_{k=1}^M f(\bar{\delta}_n, \bar{g}_k)\sum_{m=n}^N \bar{P}_{b,m}(0)\sum_{\ell=k}^M \tilde{P}_{a, \ell}, && B = 0. 
\end{aligned}
\r.
\end{align}
Given Assumption~\ref{As:QErr} and \ref{AS:EspCorr:Convex} and  using Lemma~\ref{Lem:Fun:f}, Property $1)$ in the lemma statement holds since $Z(a, b, B)$ is a nonnegative combination of monotone decreasing and convex functions of $B$ as can be observed from \eqref{Eq:J:B:1m}. Property $2)$ follows from Assumption~\ref{AS:SD}, Lemma~\ref{Lem:Fun:f} and  \eqref{Eq:J:B:1m}.  

\subsection{Proof for Theorem~\ref{Theo:Policy}}\label{App:Policy}
Since $\hat{\mathcal{P}}_{\rho}^\star\rightarrow \hat{\mathcal{P}}^\star$ as $\rho\rightarrow 1$, it is sufficient to prove that given an arbitrary $\rho \in (0, 1)$, $\hat{\mathcal{P}}_{\rho}^\star$ has the properties of $\hat{\mathcal{P}}^\star$ as described in the theorem statement. Assume that  there exists $(a, b)\in \hat{\mathcal{X}}$ such that $\hat{\mathcal{P}}_\rho^\star(a, b) = 0$. Using \eqref{Eq:Bellman} and \eqref{Eq:Z:Fun}, $(a, b)$ satisfies the following condition
\begin{equation}\label{Eq:JJCmp}
Z(a, b, 0) \leq \min_{B \neq 0} Z(a, b, B). 
\end{equation}
It follows from Assumption~\ref{AS:SD} and \eqref{Eq:J:B:1m} that with $a$ fixed, $Z(a, b, 0)$ is a monotone increasing function of $b$. As a result, we obtain from \eqref{Eq:JJCmp} that 
\begin{align}
Z(a, \delta, 0) &\leq \min_{B \neq 0} Z(a, b, B), \quad \forall \ \delta \leq b \nn\\
&= \min_{B \neq 0} Z(a, \delta, B)\label{Eq:Min:Z}
\end{align}
where \eqref{Eq:Min:Z} holds since $Z(a, b, B)$  can be observed from \eqref{Eq:J:B:1m} to be independent with $b$ if $B \neq 0$. 
Property $1)$ in the theorem statement is proved by combining \eqref{Eq:Relate:PZ} and \eqref{Eq:Min:Z}. Next, assume that  there exists $(a, b)\in \hat{\mathcal{X}}$ such that $\hat{\mathcal{P}}_{\rho}(a, b) > 0$. This implies that $\hat{\mathcal{P}}_{\rho}(a, \delta) > 0$ for 
all $\delta \geq b$ since otherwise $\hat{\mathcal{P}}_{\rho}(a, b) = 0$ based on Property~$1$, which violates the earlier assumption. Therefore, 
\begin{align}
\hat{\mathcal{P}}_{\rho}(a, \delta) &= \arg\min_{B \neq 0} Z(a, \delta, B), \quad \forall \ \delta \geq b \nn\\
&= \arg\min_{B \neq 0} Z(a, b, B)\label{Eq:Min:Z:a}
\end{align}
where \eqref{Eq:Min:Z:a} results from the equality in \eqref{Eq:Min:Z}. 
Property~$2)$ in the theorem statement follows from \eqref{Eq:Relate:PZ} and \eqref{Eq:Min:Z:a}.  

Last, assume that there exist $(\bar{g}_a, \bar{\delta}_b), (\bar{g}_c, \bar{\delta}_b) \in \hat{\mathcal{X}}$ such that $\bar{g}_a\leq \bar{g}_c$,   $\hat{\mathcal{P}}_{\rho}(\bar{g}_a, \bar{\delta}_b) > 0$, and $\hat{\mathcal{P}}_{\rho}(\bar{g}_c, \bar{\delta}_b) > 0$. To facilitate the proof, we arrange the elements of $\mathds{B}$ in the ascending order: $\mathds{B} = \{\bar{B}_1, \bar{B}_2, \cdots, \bar{B}_A\}$ with $\bar{B}_1 \leq \bar{B}_2 \cdots \leq \bar{B}_A$ and $A = |\mathds{B}|$. Moreover, given $\hat{x} \in \hat{\mathcal{X}}$, define the differences
\begin{equation}
\Delta^+Z(\hat{x}, \bar{B}_u) = Z(\hat{x}, \bar{B}_u) - Z(\hat{x}, \bar{B}_{u+1}) \label{Eq:J:Diff+}
\end{equation}
with   $1\leq u < A$ and 
\begin{equation}
\Delta^-Z(\hat{x}, \bar{B}_u) = Z(\hat{x}, \bar{B}_u) -  Z(\hat{x}, \bar{B}_{u-1}) \label{Eq:J:Diff-}
\end{equation}
with  $1 < u \leq A$. The substitution of  \eqref{Eq:J:B:1m} into  \eqref{Eq:J:Diff+} gives 
\begin{equation}\begin{aligned}
&\Delta^+ Z(\bar{g}_a, \bar{\delta}_b, \bar{B}_u) =\bar{g}_a\l\{\E[\epsilon\mid \bar{B}_u] - \E[\epsilon\mid \bar{B}_{u+1}]\r\}  +\\
&\qquad \sum_{n=1}^N \sum_{k=1}^M f(\bar{\delta}_n, \bar{g}_k) \sum_{\ell=k}^M \tilde{P}_{a, \ell}\biggl[\sum_{m=n}^N P_{b, m}(\bar{B}_u) - \\
&\qquad \qquad \sum_{m=n}^N P_{b, m}(\bar{B}_{u+1})\biggr]  - (\bar{B}_{u+1} - \bar{B}_u). 
\end{aligned}
\end{equation}
It follows that 
\begin{align}
\Delta^+& Z(\bar{g}_a,  \bar{\delta}_b, \bar{B}_u) - \Delta^+Z(\bar{g}_c, \bar{\delta}_b, \bar{B}_u) \nn\\
&= (\bar{g}_a-\bar{g}_c)\l\{\E[\epsilon\mid \bar{B}_u] - \E[\epsilon\mid \bar{B}_{u+1}]\r\}  + \nn\\
&\qquad \sum_{n=1}^N \sum_{k=1}^M f(\bar{\delta}_n, \bar{g}_k)\l[\sum_{\ell=k}^M \tilde{P}_{a, \ell} - \sum_{\ell=k}^M \tilde{P}_{c, \ell}\r]\times\nn\\
&\qquad \qquad \l[\sum_{m=n}^N P_{b, m}(\bar{B}_u) - \sum_{m=n}^N P_{b, m}(\bar{B}_{u+1})\r]\nn\\
&\leq 0\label{Eq:Delta:GZ+} 
\end{align}
where \eqref{Eq:Delta:GZ+} is obtained using Assumption~\ref{As:QErr} and \ref{AS:SD},  and Lemma~\ref{Lem:Fun:f}. Similarly, it can be shown that 
\begin{equation}
\Delta^-Z(\bar{g}_a, \bar{\delta}_b, \bar{B}_u) - \Delta^-Z(\bar{g}_c, \bar{\delta}_b, \bar{B}_u) \geq 0.  \label{Eq:Delta:GZ-} 
\end{equation}
By replacing $\bar{B}_u$ in \eqref{Eq:Delta:GZ+} and \eqref{Eq:Delta:GZ-} with $\hat{\mathcal{P}}^\star_\rho(\bar{g}_a, \bar{\delta}_b)$, 
\begin{align}
\Delta^+Z(\bar{g}_c, \bar{\delta}_b, \hat{\mathcal{P}}^\star_\rho(\bar{g}_a, \bar{\delta}_b))&\geq \Delta^+Z(\bar{g}_a, \bar{\delta}_b, \hat{\mathcal{P}}^\star_\rho(\bar{g}_a, \bar{\delta}_b))\label{Eq:ZIneq:1}\\
 \Delta^-Z(\bar{g}_c, \bar{\delta}_b, \hat{\mathcal{P}}^\star_\rho(\bar{g}_a, \bar{\delta}_b))&\leq  \Delta^-Z(\bar{g}_a, \bar{\delta}_b, \hat{\mathcal{P}}^\star_\rho(\bar{g}_a, \bar{\delta}_b)).\label{Eq:ZIneq:2}
\end{align}
 For $B\in\mathds{B}$ and $B > 0$, $Z(\hat{x}, B)$ with  $\hat{x}$ fixed is a convex function of $B$  according to Lemma~\ref{Lem:JBFunc} and from \eqref{Eq:Relate:PZ} $\hat{\mathcal{P}}^\star_\rho(\hat{x})$ minimizes $Z(\hat{x}, B)$ over $B$. Consequently, 
\begin{equation}\label{Eq:ZIneq:3}
\Delta^+Z(\hat{x}, \hat{\mathcal{P}}^\star_\rho(\hat{x}))\leq 0, \qquad \Delta^-Z(\hat{x},  \hat{\mathcal{P}}^\star_\rho(\hat{x}))\leq 0,  
\end{equation}
for any $B \leq \hat{\mathcal{P}}^\star_\rho(\hat{x})$, 
\begin{equation}
\Delta^+Z(\hat{x}, \hat{\mathcal{P}}^\star_\rho(\hat{x}))\geq 0, \qquad \Delta^-Z(\hat{x},  \hat{\mathcal{P}}^\star_\rho(\hat{x}))\leq 0, \label{Eq:ZIneq:4}
\end{equation}
and for any $B \geq \hat{\mathcal{P}}^\star_\rho(\hat{x})$, 
\begin{equation}
\Delta^+Z(\hat{x}, \hat{\mathcal{P}}^\star_\rho(\hat{x}))\leq 0, \qquad \Delta^-Z(\hat{x},  \hat{\mathcal{P}}^\star_\rho(\hat{x}))\geq 0.  \label{Eq:ZIneq:5} 
\end{equation}
It can be concluded that $\hat{\mathcal{P}}^\star_\rho(\bar{g}_c, \bar{\delta}_b)\geq \hat{\mathcal{P}}^\star_\rho(\bar{g}_a, \bar{\delta}_b)$ by combining \eqref{Eq:ZIneq:1}, \eqref{Eq:ZIneq:2} and \eqref{Eq:ZIneq:3} and comparing the result with \eqref{Eq:ZIneq:4} and \eqref{Eq:ZIneq:5}.  This proves Property $3)$ in the theorem statement, completing the proof. 

\subsection{Proof for Proposition~\ref{Prop:WaterFill:AvFb}}\label{App:WaterFill:AvFb}
We claim that there exists a threshold $\Psi:g\rightarrow\delta$ such that \eqref{Eq:OpProb:AvFb:min:a} is equivalent to the following optimization problem:
\begin{equation} \label{Eq:OpProb:AvFb:cd}
\begin{aligned}
&\textrm{minimize:}&& \E\l[ g 2^{-\frac{B}{L-1}}\mid \delta \geq \Psi(g)\r]\\
&\textrm{subject to :}&& \E[B]\leq \frac{\bar{b}}{K-1}\\
&&&  B \geq 0 
\end{aligned}
\end{equation}
and prove the claim as follows. Let $B^\star$ denote the solution of \eqref{Eq:OpProb:AvFb:min:a}. Given $g$ and from \eqref{Eq:OpProb:AvFb:min:a}, if there exists $\delta_a\in\mathcal{D}$ such that $\frac{L-1}{L}2^{-\frac{B^\star}{L-1}} < \delta_a$, $\frac{L-1}{L}2^{-\frac{B^\star}{L-1}} < \delta$ for all $\delta \geq \delta_a$; if there exists $\delta_b\in\mathcal{D}$ such that $\frac{L-1}{L}2^{-\frac{B^\star}{L-1}} \geq \delta_b$, $B^\star$ satisfies $B^\star = 0$ for all $\delta \leq \delta_b$. This proves the claim.

The above optimization problem can be solved as follows. 
First, by neglecting the positivity constraint on $B$, the convex optimization problem in \eqref{Eq:OpProb:AvFb:cd} can be solved   using Lagrangian method \cite{BoydBook}. The resultant policy is specified in \eqref{Eq:WaterFill:AvFb}. 
Next, $\Psi$ is chosen  to suppress the expected  interference power as well as enforce two constraints in 
\eqref{Eq:OpProb:AvFb:min:a}: i) feedback reduces the expected CSI error, namely $\frac{L-1}{L}2^{-\frac{B}{L-1}} < \delta$ if $B > 0$ and ii) $B \geq 0 \ \forall  \ (g, \delta)$. It follows that the problem of optimizing  $\Psi$ is as given in \eqref{Eq:OpProb:AvFb:c} with $\Psi^{-1}(1)$ replaced with $\min_{\delta}\Psi^{-1}(\delta)$. Next, it can be observed from \eqref{Eq:IPwr:Min} that given $\Pr(\delta \geq \Psi(g))$, minimizing $\bar{I}^\star$ requires $\Psi(g)$ to be a monotone decreasing function of $g$, proving the stated monotonicity of $\Psi(g)$. As a result, $\min_{\delta}\Psi^{-1}(\delta)=\Psi^{-1}(1)$ and \eqref{Eq:OpProb:AvFb:c} follows. This completes the proof.

\subsection{Proof for Lemma~\ref{Lem:Convex}}\label{App:AsymChan} Let $\mathds{P}$ denote the space of the optimal feedback-control policies. Consider $B_a, B_b \in \bigcup_{\mathds{P}} \mathcal{P}^\star(g_t^{[mn]}, \delta_t^{[mn]})$ for given $(g_t^{[mn]}, \delta_t^{[mn]})\in\mathcal{X}$. Note that  $E\l[\epsilon^{[mn]}\mid B\r] < \delta_t^{[mn]}$ if $B \in \{B_a, B_b\}$ and $B > 0$. Using this fact and for $\mu \in [0, 1]$,  
the term  $q(B) =\min\l(E\l[\epsilon^{[mn]}\mid B\r], \delta_t^{[mn]}\r)$ in the objective function of the sub-problem is proved to be convex as follows
\begin{align}
&\mu q(B_a) + (1- \mu) q(B_b)\nn \\
&= \l\{\begin{aligned}
&\!\mu E\l[\epsilon^{[mn]}_t\!\mid\! B_a\r]\! +\! (1\!-\!\mu) E\l[\epsilon^{[mn]}_t\!\mid\! B_b\r], & B_a \!>\! 0, B_b \!>\! 0\\
&\!\mu \delta_t^{[mn]}  + (1-\mu) E\l[\epsilon^{[mn]}_t\!\mid\! B_b\r], & B_a\! =\! 0, B_b\! >\! 0\\
&\mu E\l[\epsilon^{[mn]}_t\!\mid\! B_a\r]  + (1-\mu) \delta_t^{[mn]}, & B_a\! >\! 0, B_b\! =\! 0\\
&\delta_t^{[mn]}, & B_a\! =\! 0, B_b\! =\! 0
\end{aligned}  \r.  \nn\\
&\geq \l\{\begin{aligned}
&E\l[\epsilon^{[mn]}_t\mid \mu B_a + (1-\mu)B_b\r], & B_a > 0, B_b > 0\\
&E\l[\epsilon^{[mn]}_t\mid B_b\r], & B_a = 0, B_b > 0\\
&E\l[\epsilon^{[mn]}_t\mid B_a\r], & B_a > 0, B_b = 0\\
&\delta_t^{[mn]}, & B_a = 0, B_b = 0
\end{aligned}  \r. \nn
\end{align}
where the inequality uses the convexity of $E\l[\epsilon^{[mn]}_t\mid B\r]$ over $B$ as assumed in Assumption~\ref{As:QErr}. It follows that 
\begin{equation}
\mu q(B_a) + (1- \mu) q(B_b)= q(\mu B_a + (1- \mu)B_b) \nn
\end{equation}
and hence $q(B)$ is a convex function.  

Next, we prove that $\bar{I}^{[mn]}_{\min}\l(x\r)$ is a convex function for $x > 0$ using the sample-path method. Consider two average sum-feedback rates $\bar{b}_x, \bar{b}_y > 0$. Let $\mathcal{P}_x^\star$ and $\mathcal{P}_y^\star$ denote the optimal feedback-control policies that yield $\bar{I}^{[mn]}_{\min}\l(\bar{b}_x\r)$ and $\bar{I}^{[mn]}_{\min}\l(\bar{b}_y\r)$, respectively. Consider the sample paths $\l\{g_t^{[mn]}\r\}_{t=1}^\infty$ and $\l\{\delta_t^{[mn]}\r\}_{t=1}^\infty$. Let $\{B_t^x\}_{t=1}^\infty$ and $\{B_t^y\}_{t=1}^\infty$ denote the sequences of numbers of feedback bits generated by  $\mathcal{P}_x^\star$ and $\mathcal{P}_y^\star$, respectively. Moreover, given $\mu \in [0, 1]$, define the sequence $  \{B_t^z\}_{t=1}^\infty= \mu \{B_t^x\}_{t=1}^\infty + (1-\mu)\{B_t^y\}_{t=1}^\infty$. Using the function $q(B)$ defined earlier, we can write 
\begin{align}
\mu& \bar{I}^{[mn]}_{\min}\l(\bar{b}_x\r) + (1 - \mu) \bar{I}^{[mn]}_{\min}\l(\bar{b}_y\r) \nn\\
&= \lim_{T\rightarrow\infty}\frac{1}{T}\E\l[\sum_{t=1}^T \frac{1}{L-1} g_t^{[mn]} \l[\mu q(B^x_t) + (1-\mu) q(B^y_t)\r] \r]\nn\\
&\geq \lim_{T\rightarrow\infty}\frac{1}{T}\l[\sum_{t=1}^T \frac{1}{L-1} g_t^{[mn]} q(B^z_t)\r]\label{Eq:App:A2}\\
&\geq \bar{I}^{[mn]}_{\min}\l(\mu \bar{b}_x + (1-\mu)\bar{b}_y\r) \label{Eq:App:A3}
\end{align}
where \eqref{Eq:App:A2} uses  the convexity of $q(B)$ as proved earlier. The desired result follows from \eqref{Eq:App:A3}.

\begin{biography}[{\includegraphics[width=1in, clip, keepaspectratio]{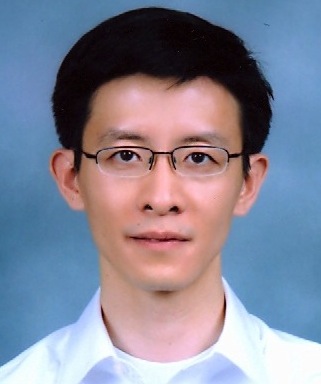}}]{Kaibin Huang}  (S'05--M'08) received the B.Eng. (first-class hons.) and the M.Eng. from the National University of Singapore in 1998 and 2000, respectively, and the Ph.D. degree from The University of Texas at Austin (UT Austin) in 2008, all in electrical engineering.

Since Mar. 2009, he has been an assistant professor in the School of Electrical and Electronic Engineering at Yonsei University, Seoul, Korea. From Jun. 2008 to Feb. 2009, he was a Postdoctoral Research Fellow in the Department of Electrical and Computer Engineering at the Hong Kong University of Science and Technology. From Nov. 1999 to Jul. 2004, he was an Associate Scientist at the Institute for Infocomm Research in Singapore. He frequently serves on the technical program committees of major IEEE conferences in wireless communications. Most recently, he is the lead publicity chair of IEEE Communication Theory Workshop 2011 and the area chairs of IEEE Asilomar 2011 and IEEE WCNC 2011. He is an editor for the Journal of Communication and Networks. Dr. Huang received  the Outstanding Teaching Award from Yonsei, the Motorola Partnerships in Research Grant, the University Continuing Fellowship at UT Austin, and the Best Student Paper award at IEEE GLOBECOM 2006. His research interests focus on limited feedback techniques  and the analysis and design of wireless networks using stochastic geometry.
\end{biography}
\vfill 

\begin{biography}[{\includegraphics[width=1in, clip, keepaspectratio]{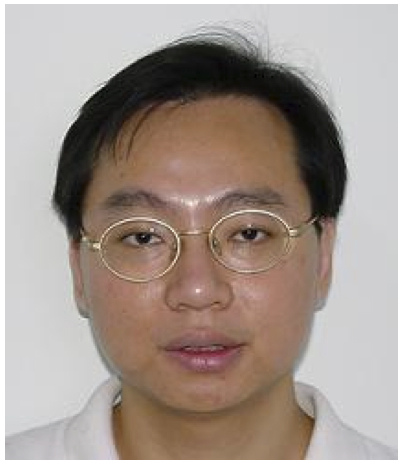}}]{Vincent K. N.  Lau} obtained B.Eng (Distinction 1st Hons) from the University of Hong Kong (1989-1992) and Ph.D. from the Cambridge University (1995-1997).  He is currently a Professor at the Dept. of ECE, Hong Kong University of Science and Technology. Vincent has contributed to more than 90 IEEE journal papers and  24 US patents on various wireless systems. His current research focus includes robust cross layer optimization for MIMO/OFDM wireless systems with imperfect channel state information, network MIMO optimization, communication theory with limited feedback as well as delay-sensitive cross layer optimizations.  
\end{biography}

\begin{biography}[{\includegraphics[width=1in, clip, keepaspectratio]{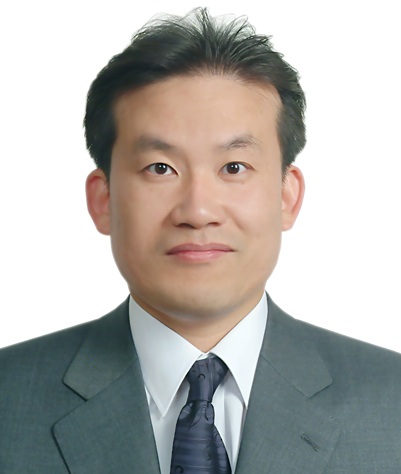}}]{Dongku Kim}   received the B. Eng. Degree from Korea Aerospace University, Korea, in 1983 and the M.Eng and Ph.D. degrees from the University of Southern California, Los Angeles, in 1985 and 1992, respectively. 

He was a Research Engineer with the Cellular Infrastructure Group, Motorola by 1994, and he has been a professor in the school of Electrical and Electronic Engineering, Yonsei University, Seoul, since 1994. He was a director of Radio Communication Research Center at Yonsei University and also a director of Qualcomm Yonsei CDMA Joint Research Lab since 1999. Currently, he is a director of Journal of Communications and Networks. His main research interests are Interference Alignment, MU-MIMO, Compressive Sensing, Mobile Multihop Relay, and UAV Tracking. 
\end{biography}

\vfill
\end{document}